\newif\ifarxiv
\arxivtrue

\ifarxiv
\documentclass[12pt,a4paper,USenglish,cleveref, autoref, thm-restate]{article}

\usepackage{amsfonts, amsmath, amssymb, amsthm}
\usepackage{mathtools}
\usepackage{enumerate}
\usepackage{cite}
\usepackage[normalem]{ulem}
\usepackage{algpseudocode}
\usepackage{algorithm}
\usepackage{xspace}
\usepackage{adjustbox}
\usepackage{comment}
\usepackage{hyperref}
\usepackage{fullpage}

\newcommand{\ov}[1]{\overline{#1}}
\newcommand{\ovar}[1]{\overrightarrow{#1}}
\newcommand{\Cell}{\mathit{Cell}}
\newcommand{\Cells}{\mathit{Cells}}
\newcommand{\Vor}{\mathrm{Vor}}
\newcommand{\Sample}{\mathit{Sample}}
\newcommand{\Piece}{\mathit{Piece}}
\newcommand{\Pieces}{\mathit{Pieces}}
\newcommand{\Sep}{\mathit{Sep}}
\newcommand{\Seps}{\mathit{Seps}}
\newcommand{\PSep}{\mathit{PSep}}
\newcommand{\Hull}{\mathit{Hull}}
\newcommand{\NN}{\mathit{NN}}
\newcommand{\old}{\mathit{old}}
\newcommand{\new}{\mathit{new}}
\newcommand{\nearest}{\mathit{nearest}}

\newcommand{\cur}{\mathtt{cur}}
\newcommand{\init}{\mathtt{init}}

\newcommand{\mcS}{\mathcal{S}}
\newcommand{\mcT}{\mathcal{T}}
\newcommand{\oT}{\overline{T}}

\newcommand{\Voron}{\mathit{Vor}}

\newcommand{\eps}{\varepsilon}

\newtheorem{theorem}{Theorem}
\newtheorem{lemma}{Lemma}
\newtheorem{corollary}{Corollary}
\DeclareMathOperator*{\argmin}{arg\,min}

\title{Incremental Planar Nearest Neighbor Queries with  Optimal Query Time} 

\author{John Iacono\thanks{Universit\'{e} libre de Bruxelles, Belgium. Research suppored by the Fonds de la Recherche Scientifique - FNRS} \and Yakov Nekrich\thanks{Michigan Technological University, USA. Supported by the National Science Foundation under NSF grant 2203278.}}
\date{}
\fi

\newif\ifsocg
\ifsocg
\documentclass[a4paper,USenglish,cleveref, autoref, thm-restate]{socg-lipics-v2021}

\usepackage{mathtools}
\usepackage{cite}
\usepackage[normalem]{ulem}
\usepackage{algpseudocode}
\usepackage{algorithm}
\usepackage{xspace}
\usepackage{adjustbox}

\newcommand{\ov}[1]{\overline{#1}}
\newcommand{\ovar}[1]{\overrightarrow{#1}}
\newcommand{\Cell}{\mathit{Cell}}
\newcommand{\Cells}{\mathit{Cells}}
\newcommand{\Vor}{\mathrm{Vor}}
\newcommand{\Sample}{\mathit{Sample}}
\newcommand{\Piece}{\mathit{Piece}}
\newcommand{\Pieces}{\mathit{Pieces}}
\newcommand{\Sep}{\mathit{Sep}}
\newcommand{\Seps}{\mathit{Seps}}
\newcommand{\PSep}{\mathit{PSep}}
\newcommand{\Hull}{\mathit{Hull}}
\newcommand{\NN}{\mathit{NN}}
\newcommand{\old}{\mathit{old}}
\newcommand{\new}{\mathit{new}}
\newcommand{\nearest}{\mathit{nearest}}

\newcommand{\cur}{\mathtt{cur}}
\newcommand{\init}{\mathtt{init}}

\newcommand{\mcS}{\mathcal{S}}
\newcommand{\mcT}{\mathcal{T}}
\newcommand{\oT}{\overline{T}}

\newcommand{\Voron}{\mathit{Vor}}

\newcommand{\eps}{\varepsilon}

\DeclareMathOperator*{\argmin}{arg\,min}

\title{Incremental Planar Nearest Neighbor Queries with  Optimal Query Time} 

\titlerunning{Incremental Planar Nearest Neighbor Queries with  Optimal Query Time} %TODO optional, please use if title is longer than one line

\author{John Iacono}{Universit\'{e} libre de Bruxelles, Belgium} {john@johniacono.com}{}{}%TODO mandatory, please use full name; only 1 author per \author macro; first two parameters are mandatory, other parameters can be empty. Please provide at least the name of the affiliation and the country. The full address is optional. Use additional curly braces to indicate the correct name splitting when the last name consists of multiple name parts.

\author{Yakov Nekrich}{Michigan Technological University, USA}{ yakov.nekrich@googlemail.com}{}{}

\authorrunning{J. Iacono and Y. Nekrich} %TODO mandatory. First: Use abbreviated first/middle names. Second (only in severe cases): Use first author plus 'et al.'

\Copyright{John Iacono and Yakov Nekrich} %TODO mandatory, please use full first names. LIPIcs license is "CC-BY";  http://creativecommons.org/licenses/by/3.0/

\ccsdesc[100]{Theory of computation~Computational geometry}
\ccsdesc[100]{Theory of computation~Data structures design and analysis}
\keywords{Data Structures, Dynamic Data Structures, Nearest Neighbor Queries} %TODO mandatory; please add comma-separated list of keywords

\EventEditors{Oswin Aichholzer and Haitao Wang}
\EventNoEds{2}
\EventLongTitle{41st International Symposium on Computational Geometry (SoCG 2025)}
\EventShortTitle{SoCG 2025}
\EventAcronym{SoCG}
\EventYear{2025}
\EventDate{June 23--27, 2025}
\EventLocation{Kanazawa, Japan}
\EventLogo{socg-logo.pdf}
\SeriesVolume{332}
\ArticleNo{57}
\fi

\newif\ifeurocg
\ifeurocg
\documentclass[a4paper,USenglish,cleveref, autoref, thm-restate]{eurocg25-submission}

\usepackage{mathtools}
\usepackage{cite}
\usepackage[normalem]{ulem}
\usepackage{algpseudocode}
\usepackage{algorithm}
\usepackage{xspace}
\usepackage{adjustbox}
\usepackage{comment}
\usepackage{enumerate}

\newcommand{\ov}[1]{\overline{#1}}
\newcommand{\ovar}[1]{\overrightarrow{#1}}
\newcommand{\Cell}{\mathit{Cell}}
\newcommand{\Cells}{\mathit{Cells}}
\newcommand{\Vor}{\mathrm{Vor}}
\newcommand{\Sample}{\mathit{Sample}}
\newcommand{\Piece}{\mathit{Piece}}
\newcommand{\Pieces}{\mathit{Pieces}}
\newcommand{\Sep}{\mathit{Sep}}
\newcommand{\Seps}{\mathit{Seps}}
\newcommand{\PSep}{\mathit{PSep}}
\newcommand{\Hull}{\mathit{Hull}}
\newcommand{\NN}{\mathit{NN}}
\newcommand{\old}{\mathit{old}}
\newcommand{\new}{\mathit{new}}
\newcommand{\nearest}{\mathit{nearest}}

\newcommand{\cur}{\mathtt{cur}}
\newcommand{\init}{\mathtt{init}}

\newcommand{\mcS}{\mathcal{S}}
\newcommand{\mcT}{\mathcal{T}}
\newcommand{\oT}{\overline{T}}

\newcommand{\Voron}{\mathit{Vor}}

\newcommand{\eps}{\varepsilon}

\DeclareMathOperator*{\argmin}{arg\,min}

\title{Incremental Planar Nearest Neighbor Queries with  Optimal Query Time} 

\titlerunning{Incremental Planar Nearest Neighbor Queries with  Optimal Query Time} %TODO optional, please use if title is longer than one line

\authorrunning{J. Iacono and Y. Nekrich} %TODO mandatory. First: Use abbreviated first/middle names. Second (only in severe cases): Use first author plus 'et al.'

\fi

\begin{document}

\ifeurocg
\author[1]{John Iacono}
\affil[1]{Université libre de Bruxelles, Belgium}%TODO mandatory, please use full name; only 1 author per \author macro; first two parameters are mandatory, other parameters can be empty. Please provide at least the name of the affiliation and the country. The full address is optional. Use additional curly braces to indicate the correct name splitting when the last name consists of multiple name parts.

\author[2]{Yakov Nekrich}
\affil[2]{Michigan Technological University, USA}%{ yakov.nekrich@googlemail.com}{}{}

\ArticleNo{100}
%\maketitle
\fi

\maketitle

\begin{abstract}
  In this paper we show that two-dimensional nearest neighbor queries can be answered in optimal $O(\log n)$ time while supporting insertions in $O(\log^{1+\eps}n)$ time.
No previous data structure was known that supports $O(\log n)$-time queries and polylog-time  insertions. 
In order to achieve logarithmic queries our data structure uses a new technique related to fractional cascading that leverages the inherent geometry of this problem. Our method can be also  used in other semi-dynamic scenarios. 
%-We describe a semi-online fully-dynamic data structure that supports queries in $O(\log n)$ time and updates in $O(\log^{1+\eps} n)$ time. \\
%-We describe an offline partially persistent data structure that can be constructed in $O(n \log^{1+\eps}n)$ time, uses $O(n\log^{1+\eps} n)$ space and supports queries in $O(\log n)$ time. 
\end{abstract}

\section{Introduction}
\label{sec:intro}
In the nearest neighbor problem a set of points $S$ is stored in a data structure so that for a query point $q$ the point $p\in S$ that is closest to $q$ can be found efficiently. The nearest neighbor problem and its variants are among the most fundamental and extensively studied problems in computational geometry; we refer to e.g.~\cite{survey1} for a survey. In this paper we study dynamic data structures for the Euclidean nearest neighbor problem in two dimensions. We show that 
the optimal $O(\log n)$ query time for this problem can be achieved while allowing insertions in time $O(\log^{1+\epsilon})$.

\paragraph*{Previous Work.}
See Table~\ref{tab:results}.
In the static scenario the planar nearest neighbor  problem can be solved in $O(\log n)$ time by point location in  Voronoi diagrams.  However the dynamic variant of this problem is significantly harder because Voronoi diagrams cannot be dynamized efficiently:  it was shown by Allen et al.\ \cite{AllenBIL17} that a sequence of insertions can lead to $\Omega(\sqrt{n})$ amortized combinatorial changes per insertion in the Voronoi diagram. A static nearest-neighbor data structure can be easily transformed into an insertion-only data structure using the logarithmic method of Bentley and Saxe~\cite{BentleyS80} at the cost of increasing the query time to $O(\log^2 n)$. Several researchers~\cite{DevillersMT92,ClarksonMS93,Mulmuley94book} studied the dynamic nearest neighbor problem in the situation when the sequence of updates is random in some sense (e.g. the deletion of any element in the data structure is equally likely). However their results cannot be extended to the case when the complexity of a specific sequence of updates must be analyzed.

Using a lifting transformation~\cite{BergCKO08}, 2-d nearest neighbor queries can be reduced to extreme point queries on a 3-d convex hulls. Hence data structures for the dynamic convex hull in 3-d can be used to answer 2-d nearest neighbor queries. The first such  data structure (without assumptions about the update sequence) was presented by Agarwal and Matou\v{s}ek~\cite{AgarwalM95}. Their data structure supports queries in $O(\log n)$ time and updates in $O(n^{\eps})$ time; another variant of their data structure supports queries in $O(n^{\eps})$ time and updates in $O(\log n)$ time. A major improvement was achieved in a seminal paper by Chan~\cite{Chan10}. The data structure in~\cite{Chan10} supports queries in $O(\log^2 n)$ time, insertions in $O(\log^3 n)$ expected time and deletions in $O(\log^6n)$ expected time. The update procedure can be made deterministic using the result of Chan and Tsakalidis~\cite{ChanT16}.  The deletion time was further reduced to $O(\log^5n)$~\cite{KaplanMRSS20} and to  $O(\log^4 n)$~\cite{Chan20a}. %We remark that it is also possible to reduce the query time to $O(\log^2n/\log\log n)$ at the cost of increasing the update time by $O(\log^{\eps}n)$ %\john{I world remove and the following this unless there is a citation or a detailed explanation}. \tcg{Agree, removed.}Thus there is a data structure that answers nearest neighbor queries in $O(\log^2n/\log\log n)$ time, insertions in $O(\log^{2+\eps}n)$ amortized time, and deletions in $O(\log^{4+\eps}n)$ amortized time. 
This sequence of papers makes use of shallow cuttings, a general powerful technique, but, alas, all uses of it for the point location problem in 2-d have resulted in $O(\log^2 n)$ query times.

Even in the case of insertion-only scenario, the direct application of the 45-year-old classic technique of Bentley and Saxe~\cite{BentleyS80} remains the best insertion-only method with polylogarithmic update before this work;
no data structure with
$o(\log^2 n)$ query time and polylogarithmic update time was described previously.

\begin{table}[]
    \centering
    \begin{tabular}{c|ccc}
         & \bf Query & \bf Insert & \bf Delete\\ \hline
Bentley and Saxe 1980 ~\cite{BentleyS80} & $O(\log^2 n)$ & $O(\log^2 n)$ & Not supported \\
Agarwal and Matou\v{s}ek  1995~\cite{AgarwalM95} &$O(\log n)$ 
& $O(n^{\eps})$ & $O(n^{\eps})$ \\
\textquotesingle \textquotesingle &$O(n^{\eps})$ 
& $O(\log n)$ & $O(\log n)$ \\
Chan 2010~\cite{Chan10}  & $O(\log^2 n)$  & $O(\log^3 n)$ \dag & $O(\log^6 n)$ \dag \\
Chan and Tsakalidis 2016~\cite{ChanT16}  & $O(\log^2 n)$ & $O(\log^3 n)$& $O(\log^6 n)$\\
Kaplan et al. 2020~\cite{KaplanMRSS20}  & 
$O(\log^2 n)$ & $O(\log^3 n)$& $O(\log^5 n)$\\
Chan 2020~\cite{Chan20a} & $O(\log^2 n)$ & $O(\log^3 n)$& $O(\log^4 n)$\\
Here & $O(\log n)$ & $O(\log^{1+\eps}n)$ & Not supported\\
\end{tabular}
    \caption{Known results. Insertion and deletion times are amortized, \dag\ denotes in expectation.}
    \label{tab:results}
\end{table}

\paragraph*{Our Results.}
We demonstrate that optimal $O(\log n)$ query time and poly-logarithmic update time can be achieved in some dynamic settings. The following scenarios are considered in this paper:
\begin{enumerate}
\item  We describe a \emph{semi-dynamic insertion-only} data structure that uses $O(n)$ space, supports insertions in $O(\log^{1+\eps}n)$ amortized time and answers queries in $O(\log n)$ time.
\item 
  In the \emph{semi-online} scenario, introduced by Dobkin and Suri~\cite{DobkinS91}, we know the deletion time of a point $p$ when a  point $p$ is inserted, i.e., we know how long a point will remain in a data structure at its insertion time.
  We describe a semi-online fully-dynamic data structure that answers queries in $O(\log n)$ time and supports updates in $O(\log^{1+\eps} n)$ amortized time. The same result is also valid in the \emph{offline} scenario when the entire sequence of updates is known in advance.
\item
  In the \emph{offline partially persistent scenario}, the sequence of updates is known and every update creates a new version of the data structure.  Queries can be asked to any version of the data structure. We describe an offline partially persistent data structure that uses $O(n\log^{1+\eps}n)$ space, can be constructed in $O(n\log^{1+\eps}n)$ time and answers queries in $O(\log n)$ time.    
\end{enumerate}

All three problems considered in this paper can be reduced to answering point location queries in (static) Voronoi diagrams of $O(\log n)$ different point sets.
For example, we can obtain an insertion-only data structure by using the logarithmic method of Bentley and Saxe~\cite{BentleyS80}, which we now briefly describe.
The input set $S$  is partitioned into a logarithmic number of subsets $S_1$, $\ldots$, $S_f$ of exponentially increasing sizes.
In order to find the nearest neighbor of some query point $q$ we locate $q$ in the Voronoi diagram of each set $S_i$ and report the point closest to $q$ among these nearest neighbors.
%Our persistent offline data structures associates a lifespan with every point.  Points are assigned to nodes of a tree $T$, in a manner similar to the segment tree. In order to answer a query, we must traverse a path in $T$ and answer a point location query in a Voronoi diagram.  The semi-online data structure uses a variant of the same approach (similar to the method introduced in~\cite{DobkinS91}). 
Since each  point location query takes $O(\log n)$ time, answering a logarithmic number of queries takes  $O(\log^2 n)$ time.   
%Answering $\tO(\log n)$ queries requires $\Omega(\log^2 n/\log\log n)$ time. 

The fractional cascading technique~\cite{ChazelleG86} applied to this problem in one dimension decreases the query cost to logarithmic by sampling elements of each $S_i$ and storing copies of the  sampled elements in other sets $S_j$, $j<i$. Unfortunately, it was shown by Chazelle and Liu~\cite{ChazelleL04} that fractional cascading does not work well for two-dimensional non-orthogonal problems, such as point location:  in order to answer $O(\log n)$ point location queries in $O(\log n)$ time, we would need $\tilde{\Omega}(n^2)$ space, even in the static scenario.

To summarize, the two obvious approaches to the insertion-only problem are to maintain a single search structure and update it with each insertion, the second is to maintain a collection of static Voronoi diagrams of exponentially-increasing size and to execute nearest neighbor queries by finding the closest point in all structures, perhaps aided by some kind of fractional cascading. The first approach cannot obtain polylogarithmic insertion time due to the lower bound on the complexity change in Voronoi diagrams caused by insertions~\cite{AllenBIL17}, and the second approach cannot obtain $O(\log n)$ search time due to Chazelle and Liu's lower bound~\cite{ChazelleL04}.
Our main intellectual contribution is showing that the lower bound of Chazelle and Liu~\cite{ChazelleL04} can be circumvented for  the case of point location in Voronoi diagrams. Specifically, a strict fractional cascading approach requires finding the closest point to a query point in each of the subsets $S_i$; we loosen this requirement: in each $S_i$, we either find the closest point or provide a certificate that the closest point in $S_i$ is not the closest point in $S$. This new, powerful and more flexible form of fractional cascading is done by  using a number of novel observations about the geometry of the problem.
 We imagine our general technique may be applicable to speeding up search in other dynamic search problems.
Our method employs planar separators to sample point sets and uses properties of Voronoi diagrams to speed up queries. We explain our method and show how it can be applied to the insertion-only nearest neighbor problem in Section~\ref{sec:basic-ins}. 
\ifarxiv
A further modification of our method that improves the insertion time and the space usage is described in Section~\ref{sec:ins-faster}. We describe a partially persistent data structure in Section~\ref{sec:persist}.  A semi-online data structure is described in Section~\ref{sec:offline}. 
\fi
\ifsocg
A further modification of our method, that will be described in the full version of this paper,   improves the insertion time and the insertion time  to $O(\log^{1+\eps}n)$ and space usage to $O(n)$. The description of partially persistent and semi-online data structures is also deferred to the full version. 
\fi

\section{Basic Insertion-Only Structure}
\label{sec:basic-ins}

We present our basic insertion only structure in several parts. In the first part, the overview (\S \ref{sec:framework}), we present the structure where one needed function, jump, is presented as a black box. With the jump function abstracted, our structure is a combination of known techniques, notably
the logarithmic method of Bentley and Saxe~\cite{BentleyS80} and sampling. In \S\ref{sec:jump} we present the implementation of the jump function and the needed geometric preliminaries. In contrast to combination of standard techniques presented in the overview which require little use of geometry, our implementation of the jump function is novel and requires thorough geometric arguments.
We then fully describe the underlying data structures needed in \S~\ref{sec:aux-datastr}. Note that all arguments presented here are done with the goal of obtaining $O(\log n)$ queries and polylogarithmic insertion. We opt here for clarity of presentation versus reducing the number of logarithmic factors in the insertion, as the arguments are already complex. 

\subsection{Overview} \label{sec:framework}

\ifarxiv
All notation is summarized in Table~\ref{table-of-notation} which can be found on the last page, page~\pageref{table-of-notation}.
We let $S$ denote the set of points currently stored in the structure, and use $n$ to denote $|S|$.
In this section we set a constant $d$ to $2$. Even though for this section $d$ is fixed, we will express everything as a function of $d$ as the techniques used in Appendix~\ref{sec:ins-faster} will use non-constant $d$.
\fi
\ifsocg
We let $S$ denote the set of points currently stored in the structure, and use $n$ to denote $|S|$. In this section we set a constant $d$ to $2$. Even though for this section $d$ is fixed, we will express everything as a function of $d$.
\fi

Let $\mcS=\{ S_1, S_2, \ldots S_{f} \}$ denote a partition of $S$ into  sets of exponentially-increasing  size where $f \coloneqq |\mcS|=\Theta(\log_d n)$ and $|S_i|=\Theta(d^i)$. Note that the partition of $\mcS$ into $S_i$ is not unique.
Let $\NN(P,q)$ be the nearest neighbor of $q$ in a point set $P$, which we assume to be unique. Given a point $q$, the computation of $\NN(S,q)$ is the query that our data structure will support.

We now define a sequence of point sets $T_1,\ldots T_{f}$. The intuition is that, as in classical fractional cascading~\cite{ChazelleG86}, the set $T_i$ contains all elements of 
$S_i$ and a sample of elements from the sets  $T_j$ where $j>i$; this implies the last sets are equal: $T_{f}=S_{f}$.
This sampling will be provided by the function $\Sample_j(k)$ which returns a subset of $T_j$ of size $O(|T_j|/d^{2k})$; while it will have other important properties, for now only the size matters. 

We now can formally define $T_i$:

$$ T_i \coloneqq S_i \cup \bigcup_{j=i+1}^{f} \Sample_j(j-i) $$
\ifarxiv
From this definition we have several observations which we group into a lemma, the proof can be found in Appendix~\ref{sec:lemma1}:
\fi
\ifsocg
From this definition we have several observations which we group into a lemma, the proof can be found in the full version of this paper:
\fi
\begin{lemma}{\bf Facts about $T_i$}  \label{lem:facts}
	\begin{enumerate}
 \setlength{\itemsep}{1pt}
\setlength{\parskip}{0pt}
\setlength{\parsep}{0pt} 
    \item $T_{f}=S_{f}$
    \item $T_i$ is a function of the $S_j$, for $j \geq i$.
    \item $S=\cup_{i=1}^{f} T_i$
    \item $\NN(S,q) \in \bigcup_{i=1}^f \{ \NN(T_i,q)\} $
    \item $|T_i| = \Theta(d^i)$
    \item For any $i$
	$ \sum_{j=i+1}^f |\Sample_j(j-i)|=\Theta(|T_i|)$
\end{enumerate}
\end{lemma}

\paragraph*{A note on notation.}
We assume the partition of $S$ into the sets $\mcS=\{S_1, S_2, \ldots , S_{f}\}$. Any further notation that includes a subscript, such as $T_i$, is a function of the $S_j$, for $j \geq i$. This compact notation makes explicit the dependence of $\mathit{anything}_i$ only on $S_j$, for $j\geq i$. This dependence in one direction only (i.e., on non-smaller indexes only) is crucial to the straightforward application of the standard Bentley and Saxe~\cite{BentleyS80} rebuilding technique in the implementation of the data structure and the insertion operation described in section \S\ref{sec:aux-datastr}-\ref{sec:insert}.

\paragraph*{Voronoi and Delaunay.}
Let $\Voron(P)$ be the Voronoi diagram of point set $P$, let $\Cell(P,p)$ be the cell of a point $p$ in $\Voron(P)$, that is the locus of points in the plane whose closest element in $P$ is $p$. Thus $q\in \Cell(P,p)$ is equivalent to $\NN(P,q)=p$.
Let $|\Cell(P,p)|$ be the complexity of the cell, that is, the number of edges on its boundary.
Let $G(P)$ refer to the Delaunay graph of $P$, the dual graph of the Voronoi diagram of $P$; the degree of $p$ in $G(P)$ is thus $|\Cell(P,p)|$ and each point in $P$ corresponds to a unique vertex in $G(P)$.  Delaunay graphs are planar.  To simplify the description, we will not distinguish between points in a planar set $P$ and vertices of $G(P)$. For example, we will sometimes say that a point $p$ has degree $x$ or say that a point $p'$ is a neighbor of $p$. We will find it useful to have a compact notation for expressing the union of Voronoi cells; thus for a set of points $P' \subseteq P$, let $\Cells(P,P')$ denote $ \bigcup_{p \in P'} \Cell(P,p)$.

\paragraph*{Pieces and Fringes.}
Given a graph, $G=(V,E)$, and a set of vertices $V' \subseteq V$, the \emph{fringe} of $V'$ (with respect to $G$) is the subset of $V'$ incident to edges whose other endpoint is in $V \setminus V'$.
Let $G=(V,E)$ be a planar graph. For any $r$, Frederickson~\cite{Frederickson87} showed the vertices of $G$ can be decomposed\footnote{We use the word \emph{decomposed} to mean a division of a set into into a collection sets, the \emph{decomposition}, whose union is the original set, but, unlike with a partition, elements may belong to multiple sets.} into $\Theta(|V|/r)$ \emph{pieces}, so that:  (i) Each vertex is in at least one piece. (ii) Each piece has at most $r$ vertices in total and  only  $O(\sqrt{r})$ vertices on its fringe. (iii) If a vertex is a non-fringe vertex of a piece (with respect to $G$), then it is not in any other pieces. (iv) The total size of all pieces is in $\Theta(|V|)$.   Intuitively, the pieces are almost a partition of $V$ where those vertices on the fringe of each piece may appear in multiple pieces.
Such a decomposition of $G$ can be computed in time $O(|V|)$~\cite{Goodrich95,KleinMS13}. 
% In \cite{KleinMS13}[Theorem 3] it was shown that for any $r$ the vertices of any planar graph $G$ with $n$ vertices can be partitioned into $\Theta(n/r)$ pieces each with at most $r$ vertices such that each piece only has $O(\sqrt{r})$ vertices on its boundary (with respect to $G$), and that such a decomposition can be computed in time $O(n)$.
We will apply this decomposition to  $T_i$, which is both a point set and the vertex set of $G(T_i)$, for exponentially increasing sizes of $r$.

Given integers $1 \leq k<j<f$, let 
\[
    \Pieces_j(k) \coloneqq \{\Piece_j^1(k), \Piece^2_j(k), \ldots  \Piece^{|\Pieces_j(k)|}_j(k)\}
\]
be a decomposition of $T_j$ into $r=\Theta(|T_j|/d^{4k})$ subsets such that each subset $\Piece^l_j(k)$  has size $O(d^{4k})$
and a fringe of size $O({d^{2k}})$ with respect to $G(T_i)$. We let \[ \Seps_j(k) \coloneqq \{\Sep_j^1(k), \Sep^2_j(k), \ldots  \Sep^{|\Pieces_j(k)|}_j(k)\}\] be defined so that 
$\Sep_j^\ell(k)$ denotes the fringe of $\Piece_j^\ell(k)$ , and let $\overline{Sep}_j^\ell(k)$ be $\Piece_j^\ell(k) \setminus \Sep_j^\ell(k)$. 
Thus each $\Piece_j^\ell(k)$  is partitioned into its fringe vertices, $\Sep_j^\ell(k)$, and its interior non-fringe  vertices $\overline{\Sep}_j^\ell(k)$; note that $\overline{\Sep}_j^\ell(k)$ may be empty if all elements of $\Piece_j^\ell(k)$ are on the fringe.

Finally, we define $\Sample_j(k)$ to be the union of all the fringe vertices: 
\[\Sample_j(k) \coloneqq \bigcup_{Sep \in \Seps_j(k)} Sep \]
thus $\Seps_j(k)$ is a partition decomposition of $\Sample_j(k)$.

For any $k\in [1.. j-1]$, the decomposition of $T_j$ into $\Pieces_j(k)$,
the partition  of each $\Piece_j^\ell(k)$ into 
$\Sep_j^\ell(k)$ and $\overline{Sep}_j^\ell(k)$, 
and the set $\Seps_j(k)$  can all be computed in time $O(|T_j|)$ using \cite{KleinMS13} if the Delaunay triangulation is available; if not it can be computed in time $O(|T_j| \log |T_j|)$. Thus computing these for all valid $i$ takes time and space $O(|T_j| \log |T_j| \log_d |T_j|)$ as $k < j=O(\log_d |T_j|)$. %\john{Explain more, reference the fact about the exponential size}
\begin{figure}[t]
    \centering
    \adjincludegraphics[width=.7\textwidth,trim={0 {.4\height} {.01\width} {.01\height}},clip,scale=.4]{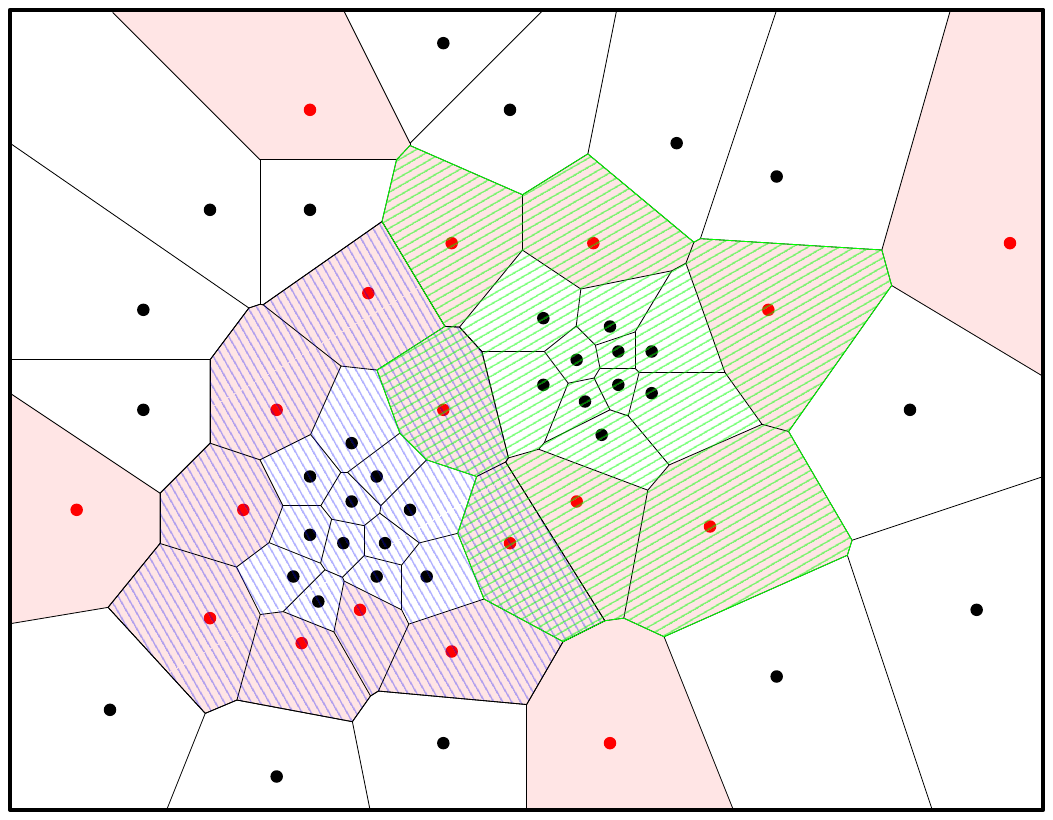}
    \caption{Part of a Voronoi diagram for a point set $T_j$.
    Two elements of $\Pieces_j(k)$ have been highlighted, one in striped blue, call it $\Piece_j^1(k)$, and one in striped green, call it $\Piece_j^2(k)$.
    For each piece, the cells of fringe vertices are shaded red. Thus, the set $\Sample_j(k)$ are the red verticies, and the region $\Cells(T_j,\Sample_j(k))$ is shaded red.
    The green-and-red shaded region is $\Cells(T_j,\Sep_j^2(k))$ and the green-but-not-red shaded region is $\Cells(T_j,\overline{\Sep}_j^2(k))$}.
    \label{fig:separ}
\end{figure}

\begin{figure}[t]
    \centering
    \adjincludegraphics[width=.5\textwidth,trim={{0.05\width} 
    {.3\height} 
    {.2\width} {.1\height}},clip,scale=.4]{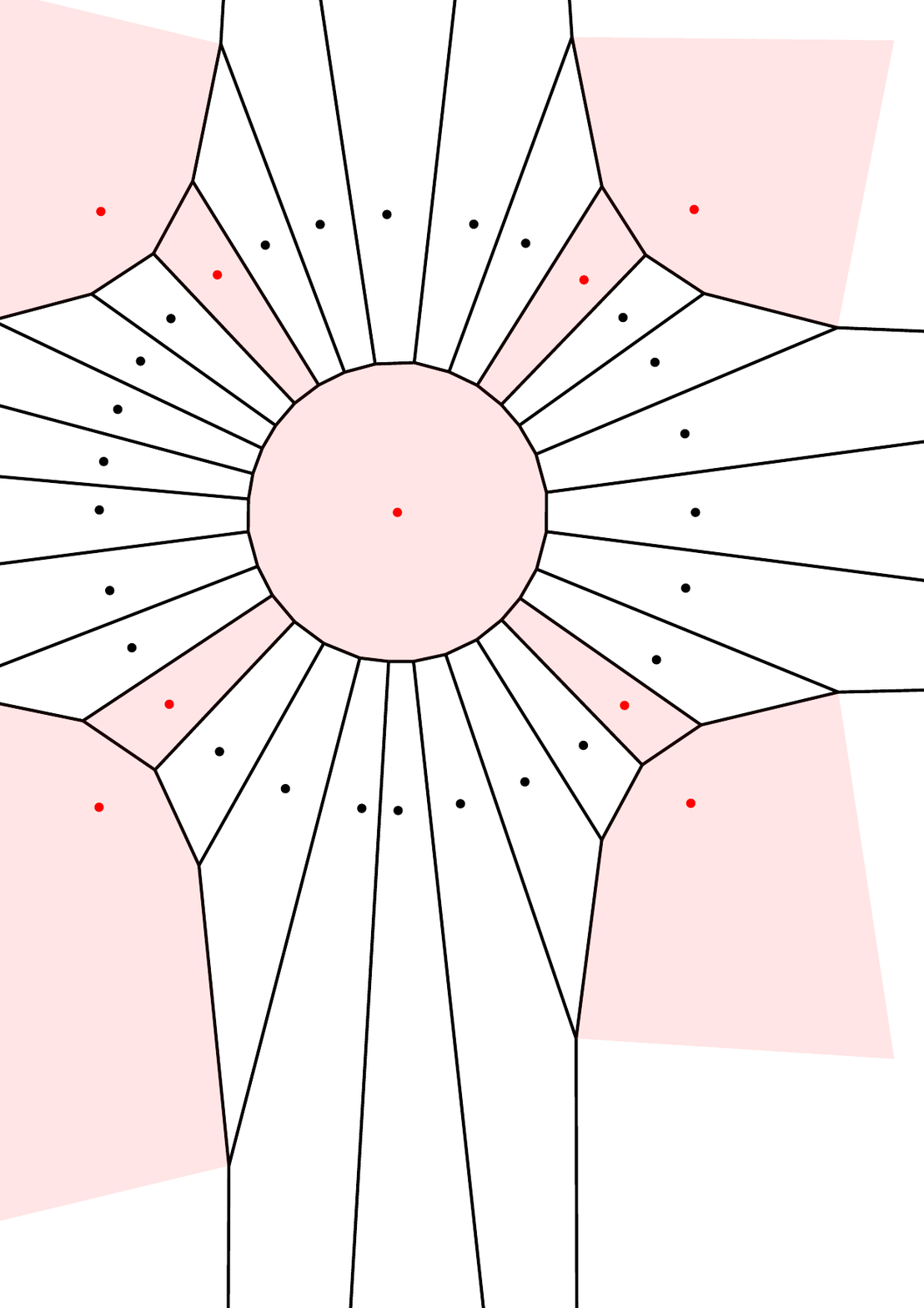}\caption{
    High complexity cells can occur in Voronoi diagrams. Such cells must be included in fringe verticies $\Sample_j(k)$, illustrated in red for some point set $T_j$. This results in the complexity of the boundary of the interior sets $\Cells(T_j,\overline{\Sep}\_j^\ell(k))$, the connected components of white Voronoi cells, are of complexity $O(d^{4(j-i)})$ by Lemma~\ref{l:inner}.}
    \label{fig:inner}
\end{figure}

One property of this sampling technique is that points in $T_j$ with Voronoi cells in $\Voron(T_j)$ of complexity at least $k$ are included in $T_i$ if $j>i$ and $j-i=O(\log_d k)$. By complexity of a region, we mean the number of edges that bound this region.

% \john{The statement here is a bit convoluted, revisit when I look at where this is used.}

%\john{Intuitively, large degree vertices in $T_j$ appear in the seperator, and then are get added to $T_i$, for some $i<j$ logarithmic in the degree. Here we us this intuition to show if a vertex is in $T_i$ and not $T_j$, its degree can be bounded.}\yakov{May be, large degree points? I suggested saying that we do not distinguish between points and vertices earlier.}

\begin{lemma}
Given $i<j$, if $p \not \in T_i$ and $p \in T_j$ then the complexity of $\Cell(T_j,p)$ is $O(d^{4(j-i)})$. 
\end{lemma}

\begin{proof}
 Suppose that $|\Cell(T_j,p)| > c d^{4(j-i)}$, for a constant $c$ chosen later, we will show this implies $p\in T_i$. Thus the degree of $p$ in $G(T_j)$ is greater than $c d^{4(j-i)}$. Consider the piece $\Piece^\ell_j(j-i)$ in $\Pieces_j(j-i)$ that contains $p$. This piece has size at most $O(d^{4(j-i)})$, which is at most $c d^{4(j-i)}$ for some $c$ (here we choose $c$).
 Thus in $G(T_j)$, $p$ must have neighbors which are not in $\Piece^\ell_j(j-i)$. By definition, $p$ is thus in the fringe $Sep^\ell_j(j-i)$, which implies $p \in \Sep^\ell_j(j-i)$. From the definition of $T_i$, $T_i \subset \Sep^\ell_j(j-i)$ and which gives $p \in T_i$. 
 \end{proof}

While we cannot bound the complexity of any Voronoi cell in a fringe, we can bound the complexity of $\overline{Sep}_j^\ell(j-i)$, the cells inside a fringe.
Intuitively, each piece has the fringe cells on its exterior and non-fringe cells on its interior; imagining the fringe cells of a piece as an annulus gives two boundaries, the exterior boundary of the fringes, which is the boundary between the cells of this and other pieces, and the interior boundary of the fringes, which is the boundary between the fringe cells and the interior cells in this piece. Crucially, while the exterior boundary could have high complexity, the interior boundary does not, which we now formalize:
\begin{lemma}\label{l:inner}
    The complexity of $\Cells(T_j,\overline{Sep}_j^\ell(j-i))$ is $O(d^{4(j-i)})$. 
\end{lemma}
\begin{proof}
See Figure~\ref{fig:inner}.
    Each cell in $\Cells(T_j,\overline{Sep}_j^\ell(j-i))$ is adjacent to either other cells in $\Cells(T_j,\overline{Sep}_j^\ell(j-i))$ or cells in $\Cells(T_j,{Sep}_j^\ell(j-i))$. The adjacency graph of these Voronoi regions is planar, and as $\overline{Sep}_j^\ell(j-i) \cup {Sep}_j^\ell(j-i)= {Piece}_j^\ell(j-i)$, and recalling that $|{Piece}_j^\ell(j-i)|=O(d^{4(j-i)})$ gives the lemma. %\john{I think this can be stated more simply, no need to mention planar, it is just the degree is at most the number of things it can be adjacant to.}\yakov{If we just sum the degrees of points in $\overline{Sep}_j^\ell(j-i)$, we will get a higher upper bound. We need planarity to state the the total number of neighbors for all points in $\overline{Sep}_j^\ell(j-i)$ is proportional to $|{Piece}_j^\ell(j-i)|$.}
\end{proof}

\paragraph*{The Jump function: definition} 

%\john{Need to add intuition here, is a bit brutal to jump into the formal definition. A figure would also be great.}

At the core of our nearest neighbor algorithm is the function $Jump$, defined as follows. 
%We now describe a function $Jump$, the implementation of which is the subject of Section~\ref{sec:jump}. 
We will find it helpful to use $\NN_R(q)$ for a range $R=[l,r]$ to denote $\NN(\cup_{i\in [l,r]}T_i,q)$; for example $\NN_{[1,k]}(q)$ is the nearest neighbor of  $q$ in $T_1,T_2,\ldots T_k$. 

Intuitively, a call to $Jump(i,j,q,p_i,e_i)$ is used when trying to find the nearest neighbor of $q$, and assuming we know the nearest neighbor of $q$ in $T_1, T_2 \ldots T_{(i+j)/2}$ seeks to provide information on whether there are any points that could be the nearest neighbor of $q$ in $T_{(i+j)/2+1} \ldots T_j$. This information could be either a simple \emph{no}, or it could provide the nearest neighbor of $q$ for some prefix of these sets. Additionally, the edge of an the Voronoi cell of the currently known nearest neighbor in the direction of the query point is always passed and returned to aide the search using the combinatorial bounds from Lemma~\ref{lem:hullbar}, point~\ref{p4}.
 
\newcommand{\jumpdef}{
\begin{itemize}
\setlength{\itemsep}{1pt}
\setlength{\parskip}{0pt}
\setlength{\parsep}{0pt} 
 %   \item We define the  function $Jump(i,j,q,p_i,e_i)$ as follows: (Yakov: I think this line is not needed)
    \item \textbf{Input to $Jump(i,j,q,p_i,e_i)$}:
        \begin{itemize}
        \setlength{\itemsep}{1pt}
     \setlength{\parskip}{0pt}
    \setlength{\parsep}{0pt} 
        \item Integers $i$ and $j$, where $j-i$ is required to be a power of 2. We use $m$ to refer to $(j+i)/2$, the midpoint.  
        \item Query point $q$.
        \item Point $p_i$ where $p_i = \NN(T_i,q)$.
        \item The edge $e_i$ on the boundary of $\Cell(T_i,p_i)$ that the ray $\ovar{p_iq}$ intersects.
  \end{itemize}
    \item \textbf{Output}: Either one of two results, \emph{Failure} or a triple $(j',p_{j'},e_{j'})$
	\begin{itemize}
		\item If \emph{Failure}, this is a certificate that $\NN_{(m,\min(j,f)]}(q) \not = \NN_{[1,\min(j,f)]}(q)$
 % \yakov{ We require that $j-i$ is a power of $2$, but $f-i$ can be not a power of $2$. Hence it is possible that $j>f$ (?).   So I think the above line should be "... for all integer $j''$, $m<j''\leq \min(j,f)$, $\NN(T_{j''},q) \not = \NN(S,q)$."}
  %\john{Thank you, fixed.}
		\item If a triple $(j',p_{j'},e_{j'})$ is returned, it has the following properties:
    \begin{itemize}
    \setlength{\itemsep}{1pt}
    \setlength{\parskip}{0pt}
    \setlength{\parsep}{0pt} 
         \item The integer $j'$ is in the range $(m,  j]$ and 
         $\NN_{(m,j')}(q) \not = \NN_{[1,j')}(q)$.
         \item The point $p_{j'}$ is $ \NN(T_{j'},q)$.
         \item The edge $e_{j'}$ is on the boundary of $\Cell(T_{j'},p_j)$ that the ray $\ovar{p_{j'}q}$ intersects.
 \end{itemize}   
 	\end{itemize}
% \item \textbf{Runtime:} The runtime of the jump function is: $O((j-i)\log d)$.
 \end{itemize}
 }
 \jumpdef
We will show later that $Jump$ runs in $O(j-i)$ time. Implementation details are deferred to Section~\ref{sec:jump}.
\paragraph*{The nearest neighbor procedure.}

%Given this jump function as a black box, we now argue that the nearest neighbor $\NN(S,q)$ can be found in $O(\log n)$ time:
A nearest neighbor query can be answered through a series of calls to the $Jump$ function: 
%\john{Which is just a more elaborate form of exponential search} \yakov{Do you mean logarithmic search?}
\begin{itemize}
 \setlength{\itemsep}{1pt}
\setlength{\parskip}{0pt}
\setlength{\parsep}{0pt}    
 \item Initialize $i=1, j=2$, $p_1$ to be $\NN(T_1,q)$, and $e_1$ to be the edge of $\Cell(T_1,p_1)$ crossed by the ray $\ovar{p_1q}$; all of these can be found in constant time as $|T_1|=\Theta(1)$. Initialize $p_{\nearest}$ to $p_1$. 
 %Initialize $maxchecked=1$; note this is for the analsysis only. \yakov{Why do we need $maxchecked$? I do not see where it is used or how it help the analysis.}
% \john{Removed. This was for an earlier version of the potential}

 \item  Repeat the following while $\frac{i+j}{2} \leq  f$:
\begin{itemize}
 \setlength{\itemsep}{1pt}
\setlength{\parskip}{0pt}
\setlength{\parsep}{0pt}    
    \item Run $Jump(i,j,q,p_i,e_i)$. If the result is failure: 
 \begin{itemize}
  \setlength{\itemsep}{1pt}
\setlength{\parskip}{0pt}
\setlength{\parsep}{0pt} 
\item Set $j=j+(j-i)$
 \end{itemize}
   \item Else a triple $(j',p_{j'},e_{j'})$ is returned: 
\begin{itemize}
     \setlength{\itemsep}{1pt}
\setlength{\parskip}{0pt}
\setlength{\parsep}{0pt} 
   % \item Set $maxchecked=j'$
            \item If $d(p_{j'},q) < d(p_{\nearest},q)$ set $p_{\nearest}=p_{j'}$
            \item Set $i=j'$ and
            %\item 
            set $j=j'+1$
\end{itemize}

\end{itemize}
\item Return $p_{\nearest}$
\end{itemize}
We will show in the rest of this section that, given this jump function as a black box, we can correctly answer a nearest neighbor query in $O(\log n)$ time. 

\begin{figure}[t]
    \centering
    \includegraphics[width=0.5\linewidth,page=2]{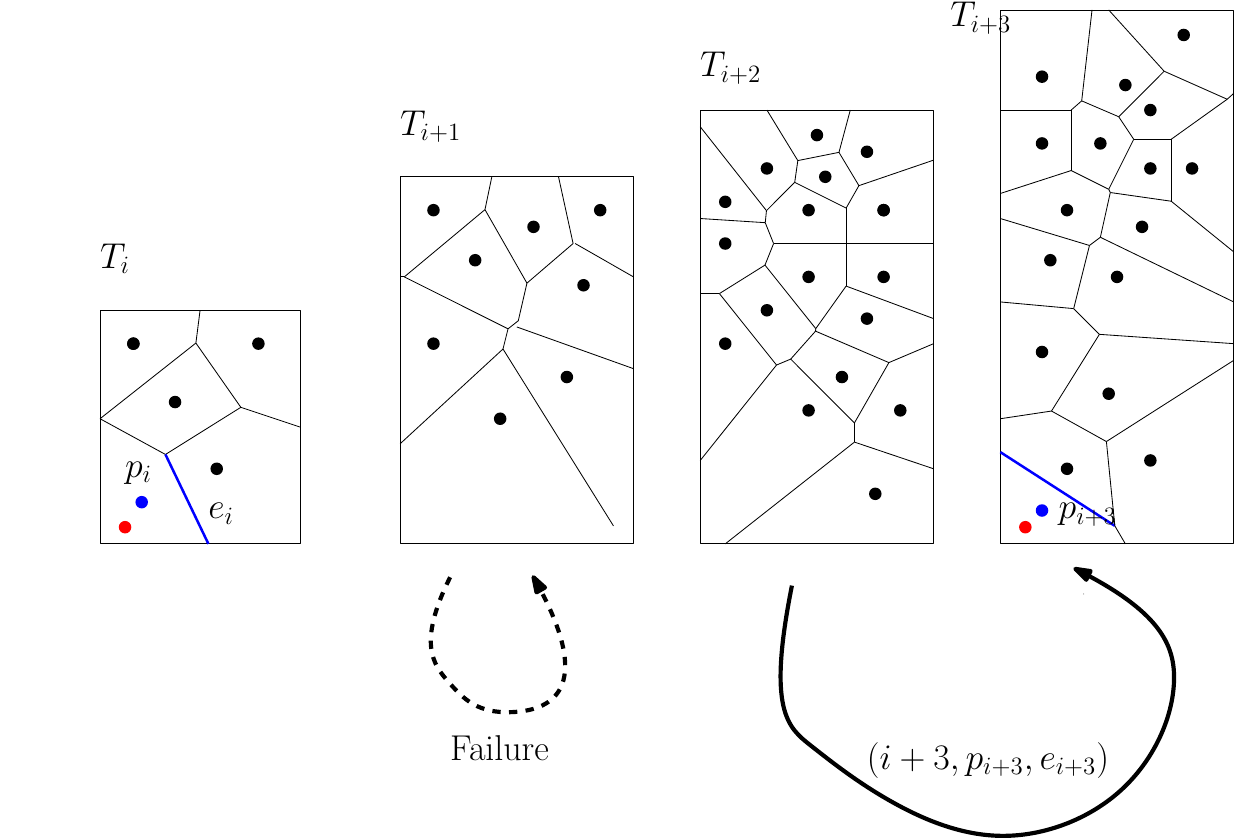}
    \caption{Two iterations of the Jump procedure. The query point $q$ is shown in red. Points $p_i=NN(T_i,q)$, $p_{i+3}=NN(T_{i+3},q)$, and edges  $e_i$ and $e_{i+3}$ are shown in blue.}
    \label{fig:jump}
\end{figure}

% \paragraph*{Alternative description of the \NN procedure}
% \yakov{I quickly wrote an alternative description, see Algorithm~\ref{alg:query}.}
% \begin{algorithm}[tb]
%   \begin{algorithmic}[1]
%     \Procedure{NearestNeighbor}{$q$}

%     \State Initialize $i\gets 1$, $j\gets 2$,
%     \State $p'_1\gets \NN(T_1,q)$
%     \State $e'_1\gets$ the edge of $\Cell(T_1,p_1)$ crossed by the ray ${p_1q}$
%     \Repeat
%         \State $result \gets Jump(i,j,q,p'_i,e'_i)$
%         \If {$result=failure$}
%         \State $j\leftarrow j+(j-i)$
%  %       \If {$j\ge\frac{3}{2}i$}
%  %       \State $i\gets \frac{i+j}{2}$
%  %       \State $j\gets i+1$
%  %      \State $p_i\gets \NN(T_i,q)$
%  %       \State $e_i\gets$ the edge of $\Cell(T_1,p_1)$ crossed by the ray $\ovar{p_iq}$
%  %       \EndIf
%         \Else \Comment{a triple $<i',j',p_{j'}>$ is returned}
%         \If {$d(p_{j'},q) < d(p_{nearest},q)$}
%         \State $p_{nearest}\gets p_{j'}$
%         \State $i\gets j'$
%         \State $j\gets j'+1$
%         \State $p'_i\gets \NN(T'_i,q)$
%         \State $e'_i\gets$ the edge of $\Cell(T'_i,p'_i)$ crossed by the ray $\ovar{p'_iq}$
%         \EndIf
%         \EndIf
%         \Until {($j>f$)}    
%   \EndProcedure  
% \end{algorithmic}
% \caption{Finding the nearest neighbor of $q$ in $\cup_{i=1}^f T_i$.}\label{alg:query}
% \end{algorithm}

\paragraph*{Correctness}
%\john{Added argument that $j-i$ is a power of 2.}
The loop in the jump procedure  will maintain the invariants that $p_{\nearest}$ is $\NN_{[1,\min (f,\frac{i+j}{2})]}(q)$, $i<j$ and $j-i$ is a power of two. The latter is because $j-i$ is set to either 1 or is doubled with each iteration of the loop.
As to the first invariant, before the loop runs, the invariant requires $p_{\nearest} =\NN(T_1,q)$, which $p_1$ and $p_{\nearest}$ are initialized to.
During the loop we subscript $i$ and $j$  by new and old to indicate the value of variables at the beginning and at the end of the loop. Thus $p_{\nearest}=\NN_{[1,\min (f,\frac{i_{\old}+j_{\old}}{2})]}(q)$. Also if the loop runs, we know $\min(j_{\old},f)=j_{\old}$.
We distinguish between two cases depending on whether the jump function returns failure or a triple.

If the jump function returns failure, we know that
$\NN_{(\frac{i_{\old}+j_{\old}}{2},\min(j_{\old},f)]}(q) \not = \NN_{[1,\min(j_{\old},f)]}(q)$. Using this fact we can go from the invariant in terms of the old variables to the new ones:

\begin{align*}
p_{nearest}&=\NN_{[1,\min(f,\frac{i_{\old}+j_{\old}}{2})]}(q)&
\text{Invariant}
\\
&=\NN_{[1,\frac{i_{\old}+j_{\old}}{2}]}(q)
& \frac{i_{\old}+j_{\old}}{2} \leq  f
\\
&=\NN_{[1,\min(j_{\old},f)]}(q)
& \NN_{(\frac{i_{\old}+j_{\old}}{2},\min(j_{\old},f)]}(q) \not = \NN_{[1,\min(j_{\old},f)]}(q)
\\
 & =\NN_{[1,\min(f,\frac{i_{\old}+j_{\old}+(j_{\old}-i_{\old})}{2})]}(q)
 &
\text{math}
\\
 & =\NN_{[1,\min(f,\frac{i_{\new}+j_{\new}}{2})]}(q)
 & i_{\new}=i_{\old};j_{\new}=j_{\old}+(j_{\old}-i_{\old})
\end{align*}

Now we consider the case when  a triple $(j',p_{j'},e_{j'})$ is returned by the $Jump$ function. We know that $\NN_{(m,j']}(q) \not = \NN_{[1,j']}(q)$, and using the same logic as from the failure case we can conclude that the old $p_{nearest}$ is
$\NN_{[1,j')}(q)$. The code sets $p_{nearest}$ to the point that is closer to $q$ among the old $p_{nearest}$, equal to $\NN_{[1,j')}(q)$, and $p_{j'}$, which is $\NN(T_{j'},q)$.
Thus the new $p_{nearest}$ is equal to $\NN_{[1,j']}(q)$ (note the closed interval) which we can rewrite to get the invariant as follows, using the fact that the subscript of $\NN$ is only dependent on the integers in the given range:
% \begin{align*}
%     \NN_{[1,j']}(q)
%     &=     \NN_{[1,i_{\new}]}(q)
% \\
%     &=     \NN_{[1,\frac{i_{\new}+i_{\new}+1}{2}]}(q)
% \\
%     &=     \NN_{[1,\frac{i_{\new}+j_{\new}}{2}]}(q)
% \end{align*}
\[    \NN_{[1,j']}(q)
    =     \NN_{[1,i_{\new}]}(q)
    =     \NN_{[1,\frac{i_{\new}+i_{\new}+1}{2}]}(q)
    =     \NN_{[1,\frac{i_{\new}+j_{\new}}{2}]}(q)
\]%
When the loop finishes, we have $j>f$ and thus
%
% \begin{align*}  
% p_{nearest} 
% &= \NN_{[1,\min (f,\frac{i+j}{2})]}(q)
% \\
% &= \NN_{[1,f]}(q)
% \\
% &= \argmin_{p\in \{ T_k|k\in [1,f]  \text{ and } k \in \mathbb{Z}\}}d(p,q)
% \\
% &= \NN(S,q)
% \end{align*}
\[p_{nearest} 
= \NN_{[1,\min (f,\frac{i+j}{2})]}(q)
= \NN_{[1,f]}(q)
= \argmin_{p\in \{ T_k|k\in [1,f]  \text{ and } k \in \mathbb{Z}\}}d(p,q)
= \NN(S,q)
\]

\paragraph*{Running time} 
\begin{lemma}
  \label{lemma:querytime}
Given the Jump function, the running time of the nearest neighbor search function in $O(\log n)$.
\end{lemma}

\begin{proof}
We use a potential argument. Let $i_t$ and $j_t$ denote the values of $i$ and $j$ after the loop has run $t$ times, let $a_t$ denote the runtime of the $t$th iteration of the loop, and let $T$ denote the number of times the loop runs. The total runtime is thus $O(1)+\sum_{t=1}^{T} a_t$. Define $\Phi_t \coloneqq c(2i_t+j_t)$, for some constant $c$ chosen so that $a_t \leq \frac{c}{2}(j_j-i_t)$; as $i_0=1$ and $j_0=2$, $\Phi_0=4c$. (recall that the runtime of Jump is $O(j-i)$ for constant $d$). 

%\yakov{You assume that $d$ is a constant? This can be extended to arbitrary $d$ by setting $\Phi_t \coloneqq c(2i_t+j_t)\log d$, right?}

We will argue that for all $t$, 
$ a_t  \leq \Phi_{t}-\Phi_{t-1} + O(1)$.
Summing, this gives $\sum_{t=1}^{T} a_t \leq \Phi_{T}-\Phi_{0}+O(T)$. Since $i$, $j$ and $T$ are in the range $[1,f]$, this bounds $\sum_{t=1}^{T} a_t$ by $3cf+O(f)=\Theta(\log n)$. %\yakov{$3cf+O(f)$?}.
Thus all that remains is to argue that $a_t \leq \Phi_{t}-\Phi_{t-1}$ for the two cases where the $t$th run of Jump ends in failure or  returns a triple.

\noindent
\textbf{Case 1, Jump returns failure.}
In the case of failure, $i$ remains the same, thus $i_t=i_{t-1}$ but $j$ increases by $j-i$, thus $j_t=2j_{t-1}-i_{t-1}$. Thus the potential increases by $c(j-i)$.

 \begin{align*}
 \Phi_t-\Phi_{t-1}&=  c(2i_t+j_t) - c(2i_{t-1}+j_{t-1})\\
 &= c(2i_{t-1}+2j_{t-1}-i_{t-1}) - c(2i_{t-1}+j_{t-1})
\\
 &= c(j_{t-1}-i_{t-1})
\\
&\geq a_t
\end{align*}

\noindent
\textbf{Case 2, Jump returns a triple.} 
%$i_{t}$ is set to $j_{t-1}$ and $j_{t}$ changes to $j_{t-1}+1$. 
%The key observation is that, due to the invariant in our correctness argument,  $j_t \geq \frac{i_{t-1}+j_{t-1}}{2}$.
$i_{t}$ is set to $j'$ and $j_{t}$ changes to $j'+1$. 
The key observation is that, due to the invariant in our correctness argument,  $j_t \geq \frac{i_{t-1}+j_{t-1}}{2}$.
Thus $i_{t}=j'\ge j_{t-1}\geq \frac{i_{t-1}+j_{t-1}}{2}$ and $j_{t}= j'+1 \ge j_{t-1}+1\geq \frac{i_{t-1}+j_{t-1}}{2}+1$
and the potential change is

 \begin{align*}
 \Phi_t-\Phi_{t-1}&=  c(2i_t+j_t) - c(2i_{t-1}+j_{t-1})\\
 &\geq  c(2\cdot\frac{i_{t-1}+j_{t-1}}{2}+\frac{i_{t-1}+j_{t-1}}{2}+1) - c(2i_{t-1}+j_{t-1})
%commented out for space
%\\
% &=  c({i_{t-1}+j_{t-1}}+\frac{i_{t-1}+j_{t-1}}{2}+1) - c(2i_{t-1}+j_{t-1})
\\
&= \frac{c}{2}(j_{t-1}-i_{t-i})+c
\\ &\geq a_t
\end{align*}
\end{proof}

\subsection{The jump function}
\label{sec:jump}

\subsubsection{Basic geometric facts}

We begin with our most crucial geometric lemma, the one that we build upon to make our algorithm work. Informally, given point sets $A$ and $B$ which possibly have elements in common, and a query point $q$, if the closest point to $q$ in $B$ is also in $A$, then the closest point to $q$ in $A \cup B$ is in $A$, and not in $B \setminus A$. 

\begin{lemma}
  \label{lemma:beats1}
  Given $i,j$, $i<j$, suppose that there is a point $q$ such that $q\in \Cell(T_i,p_i)\cap \Cell(T_j,p_j)$ for some $p_j \in \Sample_j(j-i)$. Then $q \in \Cell(T_i \cup T_j,p_i)$, or equivalently $\NN(T_i \cup T_j,q)=p_i$.\\
%  \yakov{I think it should be "Then $q \in \Cell(T_i \cup T_j,p_i)$, or equivalently $\NN(T_i \cup T_j,q)=p_i$.", right?} \john{Yes, thanks, fixed.}
\end{lemma}

\begin{proof}
  Since $p_j\in \Sample_j(j-i)$, $p_j\in T_i$ by the definition of $T_i$. 
  Since $q\in \Cell(T_i,p_i)$, $dist(q,p_i)\leq dist(q,p_j)$. 
  However, $q\in \Cell(T_j,p_j)$ and $p_j$ is the closest point to $q$ in $T_j$: $dist(q,p_j) \leq dist(q,p'_j)$ for all $p'_j$ in $T_j$. 
    Combining $dist(q,p_i)\leq dist(q,p_j)$ with $dist(q,p_j) \leq dist(q,p'_j)$ for all $p'_j$ in $T_j$ gives the lemma. 
\end{proof}

\begin{lemma} \label{lemma:convex}
    If all elements of a point set $P$ are in $\Cell(S,p)$ for some point set $S$ and $p\in S$, then all elements of the convex hull of $P$ are in  $\Cell(S,p)$ as well.
\end{lemma}

\begin{proof}
    Immediate as $\Cell(S,p)$ is convex.
\end{proof}

Lemmata~\ref{lemma:beats1} and~\ref{lemma:convex} give us a tool to determine which parts of the Voronoi cell of some $p_i$ in $\Vor(T_i)$ must also be part of the Voronoi cell of $p_i$ in $\Vor(T_i \cup T_j)$. We define this region as $\Hull_i(j,p_i)$, and then prove its properties.

Let $\Hull_i(j,p_i)$, for some $p_i \in T_i$,  denote the convex hull of
\[ \{p_i\} \cup ( \Cell(T_i,p_i) \cap   \Cells(T_j,\Sample_j(j-i)) ) .\]
%Intuitively, we consider all cells $\Cell(T_j,p)$  where $p\in \Sample_j(j-i)$ and their intersections with $\Cell(T_i,p_i)$.  
%The convex hull of these intersections and $p_i$ is $\Hull_i(j,p_i)$. 
See Figure~\ref{fig:f1} for an example.
\begin{lemma} \label{lemma:hull}
$\Hull_i(j,p_i) \subseteq \Cell(T_i \cup T_j,p_i)$ and thus if $q\in \Hull_i(j,p_i)$, $\NN(q,S) \not \in T_j \setminus T_i$
\end{lemma}
\begin{proof}
The point $p_i$ is in its own cell, $\Cell(T_i \cup T_j,p_i)$, and by Lemma~\ref{lemma:beats1}, 
all elements of $\Cell(T_i,p_i) \cap   \Cells(T_j,\Sample_j(j-i))$
are also in $\Cell(T_i \cup T_j,p_i)$. Thus the convex hull of these points is a subset of 
$\Cell(T_i \cup T_j,p_i)$ by Lemma~\ref{lemma:convex}.
\end{proof}
Thus, if a query point $q$ is in $\Hull_i(j,p_i)$, then, by Lemma~\ref{lemma:hull},  the set $T_j$ can be ignored because $dist(q,p_i)\le dist(q,p')$ for any $p'\in T_j \setminus T_i$.

Geometrically, determining whether a point in $\Cell(T_i,p_i)$ is a $\Hull_i(j,p_i)$ and thus a full search in $\Vor(T_j)$ can be skipped is what our jump function does. We now need to turn to examining the combinatorial issues surrounding $\Hull_i(j,p_i)$ and its interaction with $\Cell(T_i,p_i)$ as we need the complexity of the regions examined to be bounded in such a way to allow efficient searching to see if a point is in $\Cell(T_i,p_i)$. We begin by defining the part of $\Cell(T_i,p_i)$ that is not in $\Hull_i(j,p_i)$ as $\overline{\Hull}_i(j,p_i)$ and proving a number of properties of this possibly disconnected region.

\begin{figure}[t]
    \centering   \includegraphics[width=.99\textwidth,page=1]{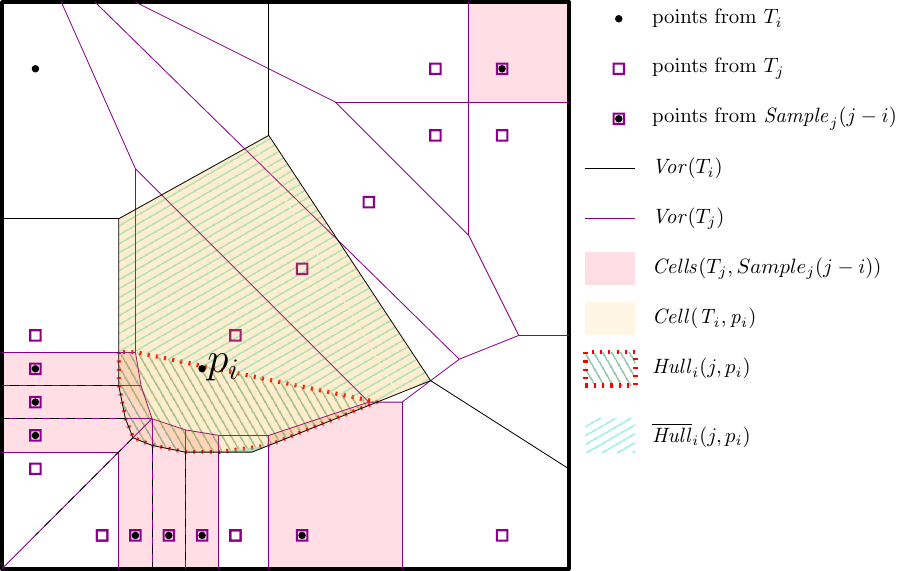}
    \caption{Illustration of the computation of $\Hull_i(j,p_i)$ and
    $\overline{\Hull}_i(j,p_i)$.
    Observe that $\Hull_i(j,p_i)$ is the convex hull of those parts of $\Cells(T_j,\Sample_j(j-i))$ (shaded pink) that are inside $\Cell(T_i,p_i)$ (shaded tan).
 $\overline{\Hull}_i(j,p_i)$ is simply the remainder of $\Cell(T_i,p_i)$, and has two connected components, including a small one at the bottom. By Lemma~\ref{lemma:hull}, the closest point in $T_i \cup T_j$ to all points in $\Hull_i(j,p_i)$ is $p_i$.
     } \label{fig:f1}
\end{figure}

%\begin{figure}
%    \centering
%    \includegraphics[width=2in,page=2]{fig/jf1-ipe.pdf}
%    \caption{Part of $\Hull_i(j,p_i)$ intersecting $\triangle p_iv_1v_2$.}
%    \label{fig:jf1}
%\end{figure}

\begin{lemma} \label{lem:hullbar}
    Consider the region 
    \[\overline{\Hull}_i(j,p_i) \coloneqq \Cell(T_i,p_i) \setminus \Hull_i(j,p_i) \]
\begin{enumerate}
    \item Each connected component %\yakov{What is a connected component? What graph do you refer to?} 
   of $\overline{\Hull}_i(j,p_i)$ is a subset of the union of Voronoi cells in one element of $\Pieces_j(j-i)$; that is, each connected component of $\overline{\Hull}_i(j,p_i)$ is a subset
    of 
    $\Cells(T_j,\Piece_j^\ell (j-i))$
    for some $\Piece_j^\ell(j-i) \in \Pieces_j(j-i)$.
    \item $\overline{\Hull}_i(j,p_i)$ intersects each bounding edge of $\Cell(T_i,p_i)$ in at most two connected components, each of which includes a vertex of $\Cell(T_i,p_i)$.
    \item Any line segment $\overline{p_iq}$, where $q$ is on the boundary of $\Cell(S_i,p_i)$ 
    intersects $\overline{\Hull}_i(j,p_i)$ in at most one connected component that, if it exists, includes $q$.
    \item \label{p:hullbar:4} Let $\overline{qr}$ be a boundary edge of $\Cell(T_i,p_i)$. The solid triangle $\triangle p_iqr$ intersects at most $d^{O(j-i)}$ edges on the boundary separating $\Hull_i(j,p_i)$ from $\overline{\Hull}_i(j,p_i)$.
\end{enumerate}
\end{lemma}

\begin{proof}
    \begin{enumerate}
        \item \label{p1} Suppose one connected component $\overline{\Hull}_i(j,p_i)$ contains points $p,p'$ in  both 
    $\Cells(T_j,\Piece^\ell_j(j-i))$ and $\Cells(T_j,\Piece^{\ell'}_j(j-i))$ where $\Piece^\ell_j(j-i)$ and $\Piece^{\ell'}_j(j-i)$ are 
        different elements of $Pieces(S_j,j-i)$. Consider a polyline that connects $p$ and $p'$ while remaining in the same connected component of $\overline{\Hull}_i(j,p_i)$. Such a polyline must cross, at some point, some cell $\Cell(T_j,p_j)$ of some $p_j \in T_j$ where $p_j \in \Piece^\ell_j(j-i)$ but 
        where the cell $\Cell(T_j,p_j)$ is adjacent to at least one other cell in $\Voron(T_j)$ that is not in $\Piece^\ell_j(j-i)$.
        Thus $p_j$ is  by definition in $\Sep_j^\ell(j-i)$; thus $p_j$ is in $\Hull_i(j,p_i)$ by its definition in Lemma~\ref{lemma:hull}. But this contradicts $p_j$ is in  $\overline{\Hull}_i(j,p_i)$.
        \item \label{p2} 
        If this does not hold, there are points $q_2$ and $q_3$ on $\overline{\Hull}_i(j,p_i)$, with $q_2$ closer to $j$ than $q_3$, such that $q_2\not\in \overline{\Hull}_i(j,p_i)$ and $q_3 \in \overline{\Hull}_i(j,p_i)$.
         But this cannot happen: We know $p\in Hull_i(j,p_i)$ by construction and if $p$ and $q_2$ are in $Hull_i(j,p_i)$, $q_2$ must be as well because the hull is a convex set.
        \item \label{p3} 
        By a similar argument as the last point, if this did not hold, there would be points $p_1,q_1,q_2,q_3$ in order on $\overline{p_iq}$, where there are some points $q_1,q_2,q_3,q_4,q_5$, in order, such that $q_1,q_3,q_5\in \overline{\Hull}_i(j,p_i)$ and $q_2,q_4 \in \Hull_i(j,p_i)$. But this cannot happen: if $q_2$ and $q_4$ are in $\Hull_i(j,p_i)$, $q_3$ must be as well because the hull is a convex set.
        \item \label{p4} 
        The complexity of $\Hull_i(j,p_i) \cap \triangle p_iqr$ is at most the complexity
        of $\Cells(T_j,\overline{Seps}_j^\ell(j-i))$
        and $\Cells(T_,\overline{Seps}_j^{\ell'}(j-i))$, where 
        $q \in \Cells(T_j,\overline{Seps}_j^{\ell}(j-i))$ 
        and
        $r \in \Cells(T_j,\overline{Seps}_j^{\ell'}(j-i))$. 
        By Lemma~\ref{l:inner}, the complexity  of these regions (within the triangle $\triangle p_iqr$) is   at most $d^{O(j-i)}$. Taking the convex hull only decreases the complexity of the objects on which a hull is defined.
        %By Lemma~\ref{l:inner}, these are of complexity at most $d^{O(j-i)}$ each.
        %These are the limits of $\Cells(T_j,\Sep_j(j-i))$ (within the triangle $\triangle p_iqr$) on which the hull is defined, taking the convex hull only decreases the complexity of the objects on which a hull is defined. 
        
    \end{enumerate}
\end{proof}

We can use these geometric facts to present the following corollary, which shows how we can use inclusion in $\Hull_i(j,p_i)$ to determine where to find the nearest neighbor of $q$ in $T_j$ or to determine that $\NN(S,q)$, $q$'s nearest neighbor in  $S$,  is not in $T_j$ \emph{without needing to find the nearest neighbor of $q$ in $T_j$}. We emphasize this is the main novel idea, since, as discussed in the introduction, finding the nearest neighbors in every $T_j$ in logarithmic time is not possible.

\begin{corollary}
    Given:
    \begin{itemize}
     \setlength{\itemsep}{1pt}
\setlength{\parskip}{0pt}
\setlength{\parsep}{0pt} 
        \item $i$ and $j$, $i<j$
        \item a query point $q$
        \item a point $p_i$ where $p_i=\NN(T_i,q)$
        \item the edge $\overline{v_1v_2}$ on the boundary of $\Cell(T_i,p_i)$ that the ray $\ovar{p_iq}$ intersects.
 \end{itemize}
 then, by testing $q$ against the part of $\Hull_i(j,p_i)$ that is inside $\triangle p_iv_1v_2$ and has complexity $d^{O(j-i)}$ one of the following is true:
\begin{itemize}
 \setlength{\itemsep}{1pt}
\setlength{\parskip}{0pt}
\setlength{\parsep}{0pt} 
\item If $q$ is inside $\Hull_i(j,p_i)$, then $q$ is in $\Cell(T_i \cup T_j,p_i)$ and $\NN(S,q)$ is not in $T_j$.
\item If $q$ is outside $\Hull_i(j,p_i)$ (and thus inside $\overline{\Hull}_i(j,p_i)$): then $\NN(T_j,q)$ is in the same element of $\overline{Seps}_j(j-i)$ as either $v_1$ or $v_2$.
\end{itemize}
\end{corollary}

\subsubsection{Implementing the jump function}

We now use one additional idea to speed up the jump function: While testing if $q$ is inside or outside of the part of $\Hull_i(j,p_i)$ that intersects $\triangle p_iv_1v_2$ can be done in time $O((j-i)\log d)$ (since the complexity of this part of the hull is $d^{O(j-i)}$ by Lemma~\ref{lem:hullbar}, point~\ref{p:hullbar:4}), we can in fact do something stronger. We can test if $q$ is inside or outside of the part of $\Hull_i(j',p_i)$ that intersects $\triangle p_iv_1v_2$, \textbf{for all $\frac{i+j}{2}<j'<j$}, in the same time $O((j-i)\log d)$. This is because the complexity of the subdivision of the plane induced by  $\frac{j-i}{2}$ hulls of size $d^{O(j-i)}$ has complexity  $d^{O(j-i)}\cdot O((j-i)^{2})$ and thus we can determine in time $O((j-i)\log d)$ which is the smallest $j'$ in the range $i<j'\leq j$ where $q$ is not inside $\Hull_i(j',p_i)$, or determine that for all $j'$ in the range $i<j'\leq j$.

Thus we obtain the final jump procedure:

% Thus we obtain the  final jump procedure, we repeat the requirements of the input and output for clarity.

%\jumpdef

\noindent \textbf{Method to compute $Jump(i,j,q,p_i,e_i)$:}
    \begin{enumerate}
        \item \label{jumpstep1} Test to see what is the smallest $j'$, $\frac{i+j}{2}< j'\leq j$ such that  $q$ is outside of the part of $\Hull_i(j',p_i)$ that intersects $\triangle p_iv_1v_2$. This will be done in time $O(\log d(j-i))$ with the convex hull search structure described in Section~\ref{sec:aux-datastr}. If it is inside all such hulls, then $\NN(S,q) \not \in T_{j'}$ for all $j'$ such that $\frac{i+j}{2}< j'\leq j$ and failure is returned. Otherwise:
        \item Let $\Piece_{v_1}$ and $\Piece_{v_2}$ denote the elements of $\Pieces_j(j'-i)$ that contain $v_1$ and $v_2$. These will be precomputed and accessible in constant time with the piece lookup  structure described in Section~\ref{sec:aux-datastr}.
        \item Search in $\Voron(T_{j'},\Piece_{v_1})$ and $\Voron(T_{j'},\Piece_{v_2})$ to find $\NN(\Piece_{v_1} \cup \Piece_{v_2},q)$, call it $p_{j'}$. Note that, $p_{j'}$ is $\NN(T_j,q)$.
         As both $\Voron(T_{j',}\Piece_{v_1})$ and $\Voron(T_{j'},\Piece_{v_2})$ have complexity $O(d^{j'-i})$, this can be done in time $O((j-i)\log d)$ 
         with the piece interior search  structure described in Section~\ref{sec:aux-datastr}.
        \item Find the edge $e_{j'}$ bounding $\Cell(T_{j'},p_{j'})$ that the ray $\ovar{p_{j'}q}$ intersects. As $\Cell(T_{j'},p_{j'})$ has complexity $d^{O(j'-i)}$, this can be done in time $O((j-i)\log d)$ with binary search.
    \end{enumerate}

\subsection{Data structures needed}
\label{sec:aux-datastr}

The data structure is split into levels, where level $i$ consists of:
\begin{enumerate}[I.]
\item $S_i$,
\item $T_i$,
\item The Voronoi diagram of $T_i$, $\Voron(T_i)$, and a point location search structure for the cells of $\Voron(T_i)$,
\item The Delaunay triangulation of $T_i$, $G(T)$.
\item Additionally, we keep for each $j$, $1\leq j <i$:
    \begin{enumerate}[i.]%[label=(V\alph*), leftmargin=10pt]
        \item The partition of $T_i$ into $\Pieces_j(k) \coloneqq \{\Piece_j^1(k), \Piece^2_j(k), \ldots  \Piece^{|\Pieces_j(k)|}_j(k)\}$
        \item The partition of each $\Piece_j^\ell(k)$ into $\Sep_j^\ell(k)$ and $\overline{Sep}_j^\ell(k)$
        \item The set $\Sample_j(k) \coloneqq \bigcup_{Sep \in \Seps_j(k)} Sep$
    \end{enumerate}
\end{enumerate}
For any level $i$, this information can be computed from $T_i$ in time $O(|T_i| \log |T_i|)$ using \cite{KleinMS13}[Theorem 3] to compute the partition into pieces, and standard results on Delaunay/Voronoi construction.

Additionally, less elementary data structures are needed for each level, which we describe separately: the convex hull search structure, the piece lookup structure, and the piece interior search structure.

\paragraph*{Convex hull search structure}

For level $i$, a convex hull structure is built for every combination of:
\begin{itemize}
 \setlength{\itemsep}{1pt}
\setlength{\parskip}{0pt}
\setlength{\parsep}{0pt} 
	\item A point $p_i$ in $T_i$
	\item An edge $e_i$ of $\Cell(T_i,p_i)$
	\item An index $j$ where $i < j$, $\frac{i+j}{2}\leq f$ and $j-i$ is a power of 2.
        \end{itemize}
        A convex hull structure answers queries of the form: given a point $q$ in $\Cell(T_i,p_i)$, return the smallest $j'$, $\frac{i+j}{2}< j'\leq j$ such that  $q$ is outside of the part of $\Hull_i(j',p_i)$ that intersects $\triangle p_iv_1v_2$, where $v_1$ and $v_2$ are endpoints of $e_i$. There are $O(|T_i| \log \log_d |T_i|)$ such structures as $j$ is at most $f=\Theta(\log_d n)$, and the complexity of a Voronoi diagram $\Voron(T_i)$ is linear in the number of points it is defined on.

The method used is to simply store a point location structure which contains subdivision within $\triangle p_iv_1v_2$ formed by the overlay of all boundaries of $\Hull_i(j',p_i)$, for $\frac{i+j}{2}< j'\leq j$. As previously mentioned the complexity of this overlay is $O(d^{j-i}\cdot(j-i)^{2})$ and thus point location can be done in the logarithm of this, which is $O((j-i)\log d)$.

 As noted earlier there are $O(|T_i| \log \log_d n)$ structures, each of which takes space at most $O(j-i)=O(\log_d n)$ plus the number of intersections in the point location structure within the triangle. The $O(j-i)$ comes from the at most 2 edges from each hull that can pass through the triangle without intersection.  For a given $j$, the sum of the complexities of $\Hull_i(j,p_i)$ over all $p_i \in T_i$ is $O(T_i)$. As each hull edge can intersect at most $O(j-i)= O(\log_d n)$ other hull edges, that bounds the total space needed over one $j$ to be $O(n \log_d n)$. The overall space usage is $O(n\log_d^2 n)$.

\paragraph*{Piece lookup structure.}
Level $i$ of the piece lookup structure contains for $j$, $i<j\leq \min(2i,f)$ and for each vertex $v_i$ of $\Voron(T_i)$ the index of which piece $\Piece^\ell_j(j-i) \in Pieces^\ell_j(j-i)$ has $q$ in $\Cells(T_j,\Piece^\ell_j(j-i))$. This can be precomputed using the point location search structure for $\Voron(T_j)$ in time $O(\log d^j)=O(j \log d)=O(i)=O(\log n)$ for fixed $d$ for each of the $O(|T_i|)$ vertices of  $\Voron(T_i)$. Summing over the choices for $j$ gives a total runtime of  $O(|T_i|\log^2 n)$ to pre-compute all answers. The space usage is $O(|T_i|\log n)$.

\paragraph*{Piece interior search structure} For each $1\leq i < j \leq f$ we store a point location structure that supports point location in time $O(j-i)$ in the Voronoi diagram for each set of points in $\Pieces_j(j-i)$. Any standard linear-sized point location structure, such as that of Kirkpatrick~\cite{Kirkpatrick83}, suffices since each element of $\Pieces_j(j-i)$ has $O(4^{j-i})$ elements. For any fixed $j$, there are $j-1$ choices for $i$, and the sets in $\Pieces_j(j-i)$ partition $T_i$. Thus the total size of all these structures is $O(|T_j| \log n)$. The construction cost, given the $\Pieces_j(j-i)$, incurs another logarithmic factor due to the need to construct the Voronoi diagram and the point location structure (we do not assume each piece is connected). Thus the piece  interior search structure for level $j$ is constructed in $O(|T_j|\log^2 n)$ time.

\subsection{Insertion time.} \label{sec:insert}

Our description of the data structures needed can be summarized as follows:

\begin{lemma}
    Level $i$ of the data structure can be built in time and space $O(|T_i| \log^2 n)$ given all levels $j>i$.
\end{lemma}

Insertion is thus handled by the classic 
logarithmic method of Bentley and Saxe~\cite{BentleyS80} which transforms a static construction into a dynamic structure, and which we briefly summarize. To insert, we put the new point into $S_1$ and rebuild level 1. Every time a set $S_i$ exceeds the upper limit of $\Theta(d^i)$, half of the items are moved from $S_i$ to $S_{i+1}$ and all levels from $i+1$ down to $S_1$ are rebuilt. So long as the upper and lower constants in the big Theta are at least a constant factor apart, the amortized insertion cost is $O(\log n)$ times the cost per item to rebuild a level, thus obtaining:

\begin{lemma}
    Insertion can be performed with an amortized running time of $O(\log^3 n)$.
\end{lemma}
The performance of our data structure can be summarized as follows:
\begin{theorem}
  \label{theor:semidyn1}
  There exists a semi-dynamic insertion-only data structure that answers two-dimensional nearest neighbor queries in $O(\log n)$ time and supports insertions  in $O(\log^3 n)$ amortized time.  The data structure uses $O(n \log^2 n)$ space. 
\end{theorem}

\bibliography{neighbor}

%\end{document}
\appendix

\newtheorem{restatedlemma}{Lemma}

\section{Proof of Lemma~\ref{lem:facts}}
\label{sec:lemma1}
\begin{restatedlemma}{\bf Facts about $T_i$}  
	\begin{enumerate}
 \setlength{\itemsep}{1pt}
\setlength{\parskip}{0pt}
\setlength{\parsep}{0pt} 
    \item $T_{f}=S_{f}$
    \item $T_i$ is a function of the $S_j$, for $j \geq i$.
    \item $S=\cup_{i=1}^{f} T_i$
    \item $\NN(S,q) \in \bigcup_{i=1}^f \{ \NN(T_i,q)\} $
    \item $|T_i| = \Theta(d^i)$
    \item For any $i$
	$$ \sum_{j=i+1}^f |\Sample_j(j-i)|=\Theta(|T_i|)$$
\end{enumerate}
\end{restatedlemma}

\begin{proof} \ \\
	\begin{enumerate}
 \setlength{\itemsep}{1pt}
\setlength{\parskip}{0pt}
\setlength{\parsep}{0pt} 
		\item $ T_f \coloneqq S_f \cup \bigcup_{j=f+1}^{f} \Sample_j(j-i) = S_f$
		\item Follows from the definition of $T_i$ and that $\Sample_j(j-i)$ is a function of $T_j$ and $j-i$ and the fact that $j>i$. 
		\item Follows from the fact that the $S_i$ are a partition of $S$ and all $T_i$ are subsets of $S$.
		\item Follows immediately from the previous point.
		\item 
  %Intuitively, the separators are geometric in size and $T_i = \Theta(S_i) = \Theta(d^i)$. 
  Since $|S_i|=\Theta(d^i)$ and $|T_f|=|S_f|=\Theta(d^f)$,
  this can be verified by solving this recurrence: %\john{Perhaps add more detail}
		$$ |T_i| \leq cd^i + \sum_{j=i+1}^f \frac{|T_j|}{d^{2(j-i)}}
		$$
		\item $$
		\sum_{j=i+1}^f |\Sample_j(j-i)|
	=
\Theta\left(		\sum_{j=i+1}^f\frac{|T_j|}{d^{2(j-i)}}	\right)
=
\Theta\left(	 |T_i|	\sum_{j'=1}^{\infty}\frac{1}{2^{j'}}	\right) =\Theta(|T_i|)
$$
	\end{enumerate}
\end{proof}

\section{Faster Updates}
\label{sec:ins-faster}
The data structure described in Theorem~\ref{theor:semidyn1} achieves optimal query time but  uses superlinear space. Superlinear space usage is needed to implement the Jump procedure: our implementation stores a poly-logarithmic number of  data items for each Voronoi edge in $T_i$. In this section we explain how the data structure can be modified so that both $O(n)$ space usage and more efficient update time are achieved.
Our improvement is based on the following idea: We maintain the Voronoi diagram for an auxiliary subset $T'_i\subset T_i$ that contains $|T_i|/\operatorname{polylog}(n)$ points.  Before we start the Jump procedure, we locate $q$ in $\Voron(T'_i)$ and use a Voronoi cell of $\Voron(T'_i)$ as the starting point of the Jump procedure.
% Our improvement  is based on the following idea: when we start the Jump procedure, we use the nearest neighbor $p'_i$ of $q$ in the subset $T'_i\subset T_i$ as the reference point. The set $T'_i$ contains $|T_i|/\operatorname{polylog}(n)$ points.
A  detailed description of the modified data structure is provided below. 

\paragraph*{Data Structure.}
We change the value of the parameter $d$ and set $d=\log^{\eps} n$ for a constant $\eps>0$. We re-define $T_i$ so that $T_i$ contains at most $\frac{|T_j|}{\log^3 n}$ points from every  set $T_j$, $j>i$: %\john{Say how this follows from what is said next}
let  $g=\log^6 n$ points, let $h=\operatorname{pow2}((1/4)\log_d g)$, and let $$ T_i \coloneqq S_i \cup \bigcup_{j=i+1}^{f} \Sample_j(\max(\operatorname{pow2}(j-i),h))$$ where $\operatorname{pow2}(x)=2^{\lceil\log_2 x\rceil}$. Thus each $T_j$ is divided into  pieces so that each piece contains $\Omega(\max(d^{4(j-i)},g))$ points and the total number of pieces is  $O(\min(\frac{|T_j|}{d^{4(j-i)}},\frac{|T_j|}{g}))$. The number of points  that are copied from $T_j$, $j>i$,  to $T_i$ does not exceed  $O(\frac{|T_j|}{\sqrt{g}})$. Hence the set  $T_i\setminus S_i$ contains $O(\frac{|T_i|\log n}{\sqrt{g}})=O(\frac{T_i}{\log^2 n})$ points. 
We define $T'_i= \Sample_h(T_i)\cup (T_i\setminus S_i)$. We construct convex hulls $\Hull_i(j,p)$ for every set $T'_i$, for every $j>i$, and for all $p\in T'_i$. We also construct all auxiliary data  structures for convex hull search described in Section~\ref{sec:aux-datastr} for every set $T'_i$. Since $T'_i$ contains $O(\frac{|T_i|}{\log^2n})$ points, all data structures can be constructed in $O(|T_i|)$ time and use $O(|T_i|)$ space. %\john{Perhaps make the time and space a lemma. Particularly the space, since this is one of the main points of this section.}

\begin{lemma}
  \label{lemma:tprime}
  If $p_i=\NN(T_i,q)$ is known, we can find $p'_i=\NN(T'_i,q)$ in $O(\log \log n)$ time. 
\end{lemma}
\begin{proof}
  For every point $p\in T_i$ we store the index $l$ such that $p_i$ is in $ \Piece^l_i(h)$. For every set
  $P^l=T'_i\cap \Piece^l_i(h)$, we store its Voronoi diagram and a point location data structure. The set $T'_i\cap \Piece^l_i(h)$ contains $\operatorname{polylog}(n)$ points because $\Piece^l_i(h)$ contains $\operatorname{polylog}(n)$ points. Hence we can find  $p^l=\NN(T'_i\cap \Piece^l_i(h),q)$ in $O(\log\log n)$ time. 

  Suppose that $p'_i\not=p^l$. Then $q$ is in $\Cell(T'_i,p')$ such that the point $p'$ is not in $\Piece^l_i(h)$. The segment $\ov{qp'}$ must intersect some $\Cell(T_i,p_s)$ such that $p_s\in \Sep^{l}_i(h)$. Consider an arbitrary planar point $q'$ on  $\ov{qp'}$ such that $q'$ is in $\Cell(T_i,p_s)$. Since $q'$ is in $\Cell(T_i,p_s)$,  $dist(p_s,q')< dist(p',q')$. On the other hand, $q'$ in $\Cell(T'_i,p')$ by convexity of Voronoi cells. Hence   $dist(p',q')< dist(p_s,q')$. A contradiction. 
Hence $p'_i=p^l$ and we can find $p'_i$  in $O(\log\log n)$ time. 
\end{proof}

\begin{lemma}
  \label{lemma:tprime2}
  Let $p_i=\NN(T_i,q)$ and $p_i'=\NN(T'_i,q)$. 
  If  the edge $v_i$ of $\Cell(T_i,p_i)$ intersected by $\ovar{p_iq}$ is known, we can find the edge $v'_i$ of $\Cell(T'_i,p'_i)$ intersected by $\ovar{p'_iq}$ in $O(\log \log n)$ time.  
\end{lemma}
\begin{proof}
  We already explained in Lemma~\ref{lemma:tprime} how $p'_i=\NN(T'_i,q)$ can be found in $O(\log\log n)$ time. We distinguish between two cases.\\
  1. $p'_i\not= p_i$.  We consider the ray $\ovar{p'_iq}$. Let $q_e$ denote the point where $\ovar{p'_iq}$ intersects  a cell $\Cell(p'_e,q)$ for some $p'_e\in T'_i$. If  $p_i$ is in $Piece^l_i(h)$, then every point $r$ on $\ov{p'_iq_e}$ is in $\Cell(T_j,p)$ for some $p\in Piece^l_i(h)$. Using the same arguments as in Lemma~\ref{lemma:tprime}, the point $r$ is in $\Cell(T'_i,p_l)$ for some $p_l\in Piece^l_i(h)\cap T'_i$.  Since the point $q_e$ is in $\Cell(T_i,p_e)$, $q_e$ is in $\Cell(T'_i,p_e)$. Hence $\ovar{p'_iq}$ intersects an edge $v'_i$ at some point $r$, such that $r$ is on $\ov{p'_iq_e}$. Therefore the edge $v'_i$ separates $\Cell(T'_i,p'_i)$ and $\Cell(T'_i,p_l)$ for some $p_l\in Piece^l_i(h)\cap T'_i$.\\
  2. $p'_i=p_i$.  We consider  the ray $\ovar{p'_iq}$. Suppose that the edge $v_i$ hit by $\ov{p'_iq}$ separates $\Cell(p'_i,T_j)$ and $\Cell(p_o,T_j)$ for some point $p_o\in Piece^{l'}_i(h)$. Let $q_e$ denote the first point where $\ovar{p'_iq}$ intersects  a cell $\Cell(p'_i,q)$ for some $p'_i\in T'_i$. Every  point $r$  on $\ov{p'_ip_e}$ in $\Cells(T_j,Piece^{l'}_i(h))$. Hence $r$ is on  $\Cells(T'_i,Piece^{l'}_i(h)\cap T'_i)$. Since $p_e$ is in $\Cell(T'_i,p_e)$, the ray  $\ovar{p'_iq}$ intersects an edge $v'_i$ at some point $r$, such that $r$ is on $\ov{p'_iq_e}$. Therefore the edge $v'_i$ separates $\Cell(T'_i,p'_i)$ and $\Cell(T'_i,p_l)$ for some $p_{l'}\in Piece^{l'}_i(h)\cap T'_i$.\\
  When the index $l$ (resp.\ $l'$) is known, we can find the edge $v'_i$ intersected by $\ovar{p'_iq}$ in $O(\log\log n)$ time by binary search. 
\end{proof}

\paragraph*{Jump Procedure.}
We modify Step~\ref{jumpstep1} of the Jump procedure. Suppose that we know the nearest neighbor $p_i$ of $q$ in $T_i$ and we know that %$\NN_{[i+1,\frac{i+j}{2}]}(q)\not= \NN_{[i,\frac{i+j}{2}]}(q)$
$dist(\NN(T_i,q),q) \ge dist(\NN_{[i,\frac{i+j}{2}]}(q),q)$ for some $j>i$.
Using Lemma~\ref{lemma:tprime}, we find $p'_i=\NN(T'_i,q)$ in $O(\log\log n)$ time. We also find the edge $e'_i=\ov{v_1v_2}$ of $\Cell(T'_i, p_i')$ that  is intersected by $\ovar{p_iq}$.  Using the hull data structure for the triangle $\triangle(p'_iv_1v_2)$, we look for the smallest $j'$, such that $\frac{i+j}{2}\le j'\le j$,  and $q$ is outside of $\Hull_i(j',p_i)$. If $j'$ is found, we execute Steps 2-4 of the Jump procedure as described in Section~\ref{sec:jump}. If there is no such hull, the procedure returns failure.  The total runtime of the Jump procedure is $O(\log\log n + \log d(j-i))=O((j-i)\log \log n)$.

Using the same analysis as in Section~\ref{sec:framework}, the total runtime of the nearest neighbor procedure is $O(\log n)$: the runtime of the Jump procedure can be bounded by  $\frac{c}{2}\log\log n (j-i)$ for some constant $c$; if we re-define the potential after the $t$-th jump to $\Phi_t\coloneqq c\log\log n(2i_t+j_t)$, then, using the same method as in Lemma~\ref{lemma:querytime}, we can bound the runtimes of all jump procedures by $\Phi_T+O(T\log\log n)$ where $T$ is the total number of jumps needed to answer a nearest neighbor query. Since $T\le f=O(\log_d n)$, the overall cost of all jumps  $O(\frac{\log n}{\log d}\log\log n)=O(\log n)$. 

\paragraph*{Piece Interior Search Structure.}
In order to complete the Jump procedure (see Step 3 in the case when the index $j'$ is found), we need to find the Voronoi cell containing $q$ provided that $\Piece^l_i(k)$ is known. This step must be completed in $O(k\log \log n)$ time. To this aim, we need an additional data structure, described in  the following lemma.
\begin{lemma}
  \label{lemma:pl-regions}
Suppose we know the piece  $\Piece^l_i(k)$  that contains the query point $q$.
Then we can find the cell $\Cell(T_i,p)$ containing  $q$ in time $O(k\cdot \log d)$. The underlying data structure uses space $O(m)$  and can be constructed in $O(m\log\log n)$ time, where $m$ is the number of points in $T_i$.
\end{lemma}
\begin{proof}
  The set $T_i$ is recursively divided into pieces $Pieces_i(k)$  where $k$ is a power of $2$ such that  $h< k  \le (\log m)/4$. 
For every such $k$ and for each $\Piece^l_i(k)$, we consider all $\Piece^{l'}_i(\frac{k}{2})$ such that $\Piece^{l'}_i(\frac{k}{2})\subset \Piece^l_i(k)$.  Let $\PSep^l_i(k)$ denote the union of all $\Sep^{l'}_i(k/2)$ such that $\Piece^{l'}_i(k/2)\subset \Piece^l_i(k)$.
For every $\Piece^l_i(k)$, we construct the Voronoi diagram of $\PSep_i^l(k)$ and a data structure that answers point location queries. The total number of points in all $P\Sep^l_i(k)$ for all $k$ and all $l$ is $O(\frac{m}{\sqrt{g}})$: For $k=2^j$, the number of points in all $P\Sep^l_i(k)$ is $O(\frac{m}{\sqrt{k}})=O(\frac{m}{d^{2^{j-1}}})$. Summing over all $k$, such that $\lceil\log h\rceil <  k < \log\log_d m$, the total number of points in all $P\Sep^l_i(k)$ is $O(\frac{m}{2^{\log h-1}})=O(\frac{m}{\sqrt{g}})$.   Hence we can construct $\Voron(\PSep^l_i(k))$ for all $\PSep^l_i(k)$ in $o(m)$ time.
Additionally, for every $\Piece^l_i(h)$, we also store its Voronoi diagram $\Voron(\Piece_i^l(h))$.  All $\Voron(\Piece_i^l(h))$ use $O(m)$ space and can be constructed in $O(m \log g)=O(m\log\log n)$ time.  

A query can be answered as follows. Suppose that $q$ is known to be in $\Piece^{l'}_i(k)$ for some $k>h$. We note that $k$ is a power of $2$.
% We find $\Piece_i^{l_0}(k_0)$ such that  $k_0=2^{\lceil \log_2 k\rceil}$ and $\Piece^{l'}_{i}(k')\subseteq \Piece^{l_0}_i(k_0)$. Then w
We set $k_0=k$, $j=0$ and repeat the following loop:
We find $\Cell(\PSep^l_j(k_j),p_{k_j})$ that contains $q$ and test whether $q$ is in $\Cell(T_i,p_{k_j})$. If this is the case, we stop. Otherwise we  find the edge $e'$ that is hit by the ray $\ovar{p_{k_j}q}$.  hen we increment $j$ and  set $k_j=k_{j-1}/2$. When $e'$ is known, we can identify 
$\Piece^{l_j}_i(k_{j})$ that contains the query point $q$ and start the next iteration of the loop.
If  $k_j=h$, we stop the loop and find $\Cell(\Piece^{l_j}_i(h), p)$ that contains $q$. The $j$-th iteration of the loop takes $O(\log(d^{k_j}))=O(k_j\log d)$ time. Hence a query is answered in $O(k_0\log d)=O(k\log \log n)$ time.
\end{proof}

\paragraph*{Insertions.}
Now we can evaluate the overall cost of constructing the level $i$ of our data structure. The set $T_i$ is updated by selecting a subset $\oT\subset T_{i-1}$, such that $\oT$ contains $d^{i-1}/2$ points, and moving  points  $p\in \oT$ from $T_{i-1}$ to $T_i$.  We can choose the points of $\oT$; for example, we can select $d^{i-1}/2$ points from $T_{i-1}$ with the smallest $x$-coordinates. Hence we can obtain the Voronoi diagram of $\oT$ in $O(|T_i|)$ time. We can also merge the Voronoi diagrams of $\oT$ and $T_i$ and obtain the Voronoi diagram of the updated set $T_i$ in $O(|T_i|)$ time using a linear-time  algorithm for merging Voronoi diagrams~\cite{Kirkpatrick79,Chazelle92}. When $\Voron(T_i)$ is available, we obtain $\Sample_i(k)$ for all $k\le (1/4)\log_d n$ such that $k$ is a power of $2$. Each $\Sample_i(k)$ can be obtained in $O(|T_i|)$ time. Hence we need $O(|T_i|\log\log n)$ time to obtain all $\Sample_i(k)$ for  $k\le (1/4)\log_d n$ such that $k$ is a power of $2$.  Then we re-build Voronoi diagrams for $T_j$, $j<i$: for every $j<i$, we construct the Voronoi diagram of $T^{\mathrm{aux}}_j=\cup_{j'>j} \Sample_{j'}(\operatorname{pow2}(\max(j'-j,h)))$. Since this set contains $O(|T_j|/g)$ points, we can construct the Voronoi diagram of $T^{\mathrm{aux}}_j$ in $O(|T_j|)$ time. Finally we merge  $\Voron(T^{\mathrm{aux}}_j)$  and the Voronoi diagram of $S_j$  and obtain $\Voron(T_j)$ in $O(|T_j|)$ time.

When $\Voron(T_i)$ is available, we can extract the points of $T'_i$, construct $\Voron(T'_i)$ and all convex hull structures in $O(|T_i|)$ time. We can also construct the piece interior search structure for $T_i$ in $O(|T_i|\log\log n)$ time. For each
$j<i$, we can construct $T'_j$ and auxiliary data structures in $O(|T_j|\log \log n)$ time. In summary, we can re-build the level $i$ and levels $j<i$ of our data structure in time $O(|T_i|\log\log n)$.

Using the standard logarithmic method analysis, the overall insertion time is $O(\log n\cdot d \cdot \log\log n)=O(\log^{1+\eps}n\log\log n)$ time. If we replace $\eps$ with $\eps'<\eps$ in the above description,  the insertion time is reduced to $O(\log^{1+\eps'}n\log\log n)=O(\log^{1+\eps}n)$. We obtain the following result.
\begin{theorem}
  \label{theor:semidyn2}
  There exists a semi-dynamic insertion-only  data structure that answers two-dimensional nearest neighbor queries in $O(\log n)$ time and supports insertions  in $O(\log^{1+\eps} n)$ amortized time.  The data structure uses $O(n)$ space. 
\end{theorem}

\section{Offline Persistent Data Structure}
\label{sec:persist}
In this section we explain how our method can be used in the offline partially persistent fully dynamic scenario, i.e., in the case when the entire sequence of updates is known in advance and we can ask nearest neighbor queries to any version of the data structure.

We associate a \emph{lifespan} with every point $p$. If $p$ was inserted at time $t_1$ and removed at time $t_2$, then the lifespan of $p$ is $[t_1,t_2]$; we assume $t_2=\infty$ if $p$ is never removed.  All points are stored in a variant of the segment tree data structure. The leaves of the segment tree $\mcT$ store the insertion and deletion times of different points in sorted order. Each internal node has $\log^{\eps} n$ children. The range $rng(u)$ of a node $u$ is the interval $[u_{\min},u_{\max}]$ where $u_{\min}$ ($u_{\max}$) is the value stored in the leftmost (rightmost) leaf descendant of $u$.  We store a point $p$ in the set $S(u)$ if $rng(u)\subseteq [t_{1}(p),t_{2}(p)]$, but $rng(\text{parent}(u))\not\subseteq [t_{1}(p),t_{2}(p)]$ where $[t_{1}(p),t_{2}(p)]$ is the lifespan of the point $p$. Each point is stored in $O(\log^{1+\eps} n)$ sets $S(u)$. If the lifespan of a point $p$ includes time $t$, then there is exactly one node $u$ such that $p\in S(u)$ and $u$ is an ancestor of the leaf $\ell_t$ that holds the value $t$.  In order to find the nearest neighbor of a point $q$ at time $t$ we must find the nearest neighbor of $q$ in $\cup_{u\in \pi(t)}S(u)$ where the union is taken over all nodes on the path $\pi(t)$ from the root to $\ell_t$.  

Following the approach of Sections~\ref{sec:basic-ins} and~\ref{sec:ins-faster}, we define $T(u)\supseteq S(u)$ for all nodes of $\mcT$. For the root node $u_R$, $T(u_R)=S(u_R)$. For a non-root node $u$, we denote by $\text{anc}_j(u)$ the height-$j$ ancestor of $u$ and let  $T_j(u)=T(\text{anc}_j(u))$. Sets $\Sample_j(k,u)$ are defined with respect to $T_j(u)$ in the same way as in Section~\ref{sec:framework}. We set $T(u)=S(u)\cup_{j>\text{height}(u)} \Sample_j(j-i,u)$. Thus sets $T(u)$ can be constructed for all nodes $u\in \mcT$ in top-to-bottom order. When $T(u)$ is defined, we can construct $T'(u)$ and all auxiliary data structures necessary for the Jump procedure. Now we can find $\NN(\cup _{u\in \pi(t)}S(u),q)=\NN(\cup _{u\in \pi(t)}T(u),q)$ as follows. We start in a leaf node of $\pi(t)$ and move up along $\pi(t)$ using the Jump procedure.  Using the analysis of Lemma~\ref{lemma:querytime}, the total cost of all Jumps is $O(\log n)$.

% We define $DOWN(S'(u),j)=\Sub(S'(u),j)$. Let $S(v)=S'(v)\cup (\cup_{u\in\text{anc}(v)}DOWN(u,\text{dist}(u,v)))$, where $v$ is any node in $\mcT$, $\text{anc}(v)$ is the set of its proper ancestors, and $\text{dist}(u,v)$ denotes the length of the path from $u$ to $v$. For any root-to-leaf path $\pi$,  $\cup_{u\in \pi}S'(u)= \cup_{u\in \pi}S(u)$. Using the method described in Sections~\ref{sec:basic-ins} and-~\ref{sec:ins-faster}, we can find the nearest neighbor of a query point $q$ in $\cup_{u\in \pi}S(u)$ in $O(\log n)$ time. The difference in the implementation of the query algorithm: there is no lower bound on the size of a set $S(u)$ stored in the node $u$.
% Therefore we must store in every node $u$ the maximum length of a jump that is possible from a point $p\in \oS(u)$. Every jump procedure from a point $p\in \oS(u)$ to its ancestor $v$ takes time $O(\text{dist}(u,v)\cdot \log\log n)$. Hence a query is answered in $O(\log n)$ time. 
\begin{theorem}
  \label{theor:persist}
  There exists an offline partially persistent data structure that answers two-dimensional nearest neighbor queries in $O(\log n)$ time. The data structure uses $O(n \log^{1+\eps} n)$ space and can be constructed in $O(n\log^{1+\eps}n)$ time. 
\end{theorem}

\section{Semi-Online Fully-Dynamic Data Structure}
\label{sec:offline}
In this section we consider the ephemeral (non-persistent) offline  scenario. In this scenario the sequence of all insertions and deletions is known in advance. However only the latest version of the data structure can be queried.

We order updates and assign every update an integer timestamp. As in Section~\ref{sec:persist} we associate a lifespan $[t_1(p),t_2(p)]$ with every point $p$; we also maintain a tree $\mcT$, so that the leaves of $\mcT$ store the insertion and deletion times of different points in sorted order.
The definition of a set $S(u)$ is changed so that every point is stored in only one set at every time.
Let $\cur$ denote the timestamp of the current version.  The set $S(u)$ contains all points $p$ such that
$rng(u)\subseteq [t_{1}(p),t_{2}(p)]$, $rng(\text{parent}(u))\not\subseteq [t_{1}(p),t_{2}(p)]$,
and $u_{\min}\le \cur \le u_{\max}$.

After every update, we modify our data structure as follows.
We examine sets $S(u)$ for every node $u$, such that $u_{\max}=\cur$. For every such node $u$ and for each $p\in S(u)$, we remove $p$ from $S(u)$ and insert it into $S(u')$ such that
$rng(u')\subseteq [t_{1}(p),t_{2}(p)]$, $rng(\text{parent}(u'))\not\subseteq [t_{1}(p),t_{2}(p)]$,
and $u'_{\min}\le \cur+1 \le u'_{\max}$.  If such $u'$ does not exist, we remove $p$ from our data structure.
When all sets $S(u)$ are processed, we increment $\cur$ by $1$.
Then we examine all nodes $v$, such that $v_{\min}=\cur$ in the decreasing order of their height.  For each $v$ we 
construct $T(v)=S(v)\cup (\cup_{j>\text{height}(u)} \Sample_j(j-i,u))$. We also construct  and store all auxiliary data structures.

Every point $p$ is inserted into $O(\log^{1+\eps} n)$ different sets because there are $O(\log n)$ nodes $u$, such that
$rng(u)\subseteq [t_{1}(p),t_{2}(p)]$ but $rng(\text{parent}(u))\not\subseteq [t_{1}(p),t_{2}(p)]$.  Every set $S(u)$ is created only once.  Hence the amortized update time is $O(\frac{c(n)}{n}\log^{1+\eps} n)=O(\log^{1+2\eps}n)$, where $c(n)$ is the construction time of the data structure. By replacing $\eps$ with $\eps/2$, we obtain the following result.  
\begin{theorem}
  \label{theor:offline}
  There exists a fully-dynamic offline data structure that answers two-dimensional nearest neighbor queries in $O(\log n)$ time. The data structure uses $O(n)$ space where $n$ is the total number of updates. Each update operation is supported in $O(\log^{1+\eps} n)$ amortized time.  
\end{theorem}

The data structure of Theorem~\ref{theor:offline} can be adjusted to the semi-online scenario.  In this setting the sequence of updates is not known in advance. But we know the deletion time $t_2(p)$ of each point $p$ at the time when $p$ is inserted. We initially create  the tree $\mcT$ with $n'=2n_0$ leaves where $n_0$ is the  number of points initially stored in the data structure. When a new point with lifespan $[t_1(p),t_2(p)]$ is inserted, we set $t'_1(p)=t_1(p)$ and $t'_2(p)=\min(t_2(p),\init +n'-1)$. The variable $\init$ is set to $1$ when the data structure is initialized and updated every time when we re-build the tree $\mcT$.  We insert $p$ with lifespan $[t'_1(p),t'_2(p)]$ into $\mcT$  as described in the proof of Theorem~\ref{theor:offline}.

After $n_0/2$ insertions or deletions, we set $n_0$ to the current number of elements in the data structure and  build a new tree with $n'=2n_0$ leaves.  All points currently stored in the nodes of the old tree are moved to the new tree. For each point $p$, we set $t'_{1}(p)=\cur$ and $t'_{2}(p)=\min(t_2(p),\cur+n'-1)$. For simplicity, all leaves in the new tree are indexed with $\cur$, $\cur+1$, $\ldots$, $\cur+n'-1$.  Finally, we set $\init= \cur$ and discard the old tree.  The global re-building incurs $O(1)$ insertions into the new tree per update.  Hence, the amortized update cost is unchanged.
\begin{corollary}
  \label{cor:semi-online}
  There exists a fully-dynamic semi-online data structure that answers two-dimensional nearest neighbor queries in $O(\log n)$ time. The data structure uses $O(n)$ space where $n$ is the total number of updates. Each update operation is supported in $O(\log^{1+\eps} n)$ amortized time.
\end{corollary}

%\pagebreak

\begin{table} \small 
\begin{tabular}{|c|p{3.5in}|}\hline
$P$ & A generic set of points with no specific meaning \\
$S$ & The set of points currently stored \\
$n$ & $|S|$ \\
$dist(p,q)$ & Distance from point $p$ to $q$\\
$\triangle pqr$ & Triangle with endpoints $p$, $q$, and $r$\\
    &       \\
$d$ & A parameter initially set to a constant\\
${\cal S}=\{ S_1, S_2, \ldots S_{f} \}$ & A partition of $S$, where $|S_i|=\Theta(d^i)$ \\  
$f$ & $|{\cal S}|$\\
$T_i$ & $S_i \cup \bigcup_{j=i+1}^{|\mcS|} \Sample_j(j-i)$\\
      &    \\
$\Voron(P)$ & The Voronoi diagram of $P$ \\
$G(P)$ &  The Delaunay graph of $P$ \\
$\Cell(P,p)$ & The cell of $p$ in $\Voron(P)$ \\
$\Cells(P,P')$ & For $P' \subseteq P$, the union of $\Cell(P,p)$ for all $p \in P'$. \\
$\NN(P,q)$ & The nearest neighbor of $q$ in $P$.\\
$\NN_R(q)$ & $\argmin_{p\in \{ T_k|k\in R  \text{ and } k \in \mathbb{Z}\}}dist(p,q)$\\
$\begin{array}{c}
\Pieces_j(k) = \\ \{\Piece^1_j(k),  \ldots, \Piece^{|\Pieces_j(k)|}_j(k)\}
\end{array}$
& A division of $T_j$ into $\Theta(|T_j|/d^{4k})$ subsets such that each subset has size $O(d^{4k})$.
\\
$\Sep^\ell_j(k)$ &  The subset of points in $\Piece^\ell_j(k)$ who have neighbors in $G(T_j)$ that are not in $\Piece^\ell_j(k)$.  
\\
$\overline{Sep}_j^\ell(k)$ & $\Piece_j^\ell(k) \setminus \Sep_j^\ell(k)$
\\
$\Seps_j(k)$ & $\{\Sep_j^1(k), \Sep^2_j(k), \ldots  \Sep^{|\Pieces_j(k)|}_j(k)\}$
\\
$\Sample_j(k)$ & $\bigcup_{Sep \in \Seps_j(k)} Sep$ \\
%$T_i$ & $S_i \cup \bigcup_{j=i+1}^{|\mcS|} \Sep_j(j-i)$
%\\
%$far(S_i,S_j,p_i,j-i,\theta)$ & The last intersection point of the ray from $p_i$ at an angle of $\theta$ with an element of $\Cell(S_i,p_i) \cap  \bigcup_{p_j \in Sep(S_j,j-i)} \Cell(S_j,p_j)$.
%\\
%$far(T_i,T_j,p_i,j-i)$ & The polyline formed by $far(S_i,S_j,p_i,j-i,\theta)$ for all $\theta$\\
$\Hull_i(j,p_i)$, $p_i \in T_i$ & The convex hull of\\
 & \hspace{.5cm} $ \{p_i\} \cup ( \Cell(T_i,p_i) \cap   \Cells(T_j,\Sep_j(j-i)) ) $
\\
$\overline{\Hull}_i(j,p_i)$, $p_i \in T_i$ & 
$\Cell(T_i,p_i) \setminus \Hull_i(j,p_i)$
\\ \hline
%$boundary(S_i,p_i,q)$ & The intersection of ray $\ovar{p_iq}$ and the boundary of $\Cell(S_,p_i)$
%\\
%$edge(S_i,p_i,q)$ & The edge of $\Cell(S_,p_i)$ that contains $boundary(S_i,p_i,q)$
\end{tabular}

\caption{
    Table of notation. }
    \label{table-of-notation}
    %\\ \yakov{$\Cell(P,p)$ denotes a Voronoi cell and later, in the definition of $hull_i(j,p_i)$,  $\Cell(P,p)$ also denotes the  set of  vertices of a Voronoi cell.} \john{$\Cell(P,p)$ is a region, that is an infite set of points, thus it can be used in hull.}}
\end{table}

\end{document}